\documentclass{CSML}

\def\dOi{13(2:13)2017}
\lmcsheading%
{\dOi}
{1--25}
{}
{}
{May\phantom.~\phantom08, 2014}
{Jun.~30, 2017}
{}

\usepackage{amsmath,amsfonts,amssymb,latexsym}
\usepackage{stmaryrd,proof}
\usepackage{hyperref}\hypersetup{hidelinks}
\usepackage{todonotes}
\usepackage{algorithm,algorithmic}
\usepackage{multirow,wrapfig} 

\usepackage{tikz}
\usetikzlibrary{calc}
\usetikzlibrary{shapes.multipart}
\usetikzlibrary{matrix}

\usepackage{xifthen}
\providecommand*{\ifempty}[3]{\ifthenelse{\isempty{#1}}{#2}{#3}}

\newcommand{\N}[1][]{\mathbb{N}_{#1}}
\newcommand{\Q}[1][]{\mathbb{Q}_{#1}}
\newcommand{\R}[1][{\geq 0}]{\mathbb{R}_{#1}}
\newcommand{\qee}{\hfill{\scriptsize$\blacksquare$}}

\newcommand{\M}{\mathcal{M}}  
\newcommand{\C}{\mathcal{C}} 
\newcommand{\Dist}{\mathcal{D}}
\newcommand{\norm}[3][]{ \| #2 - #3 \|_{#1} }
\newcommand{\tv}{\norm[\textsc{TV}]}
\newcommand{\lbld}[2]{\ifempty{#1#2}{\mathcal{L}}{\mathcal{L}(#1,#2)}}
\newcommand{\Exp}{\mathsf{exp}}
\newcommand{\rated}[2]{\ifempty{#1#2}{\mathcal{E}}{\mathcal{E}(#1,#2)}}
\newcommand{\trd}[3][d]{\ifempty{#1#2#3}{\mathcal{T}}{\mathcal{T}(#1)(#2,#3)}}
\newcommand{\trcd}[3][d]{\ifempty{#1#2#3}{\Theta}{\Theta(#1)(#2,#3)}}

\newcommand{\set}[3][]{#1\{ #2  \ifempty{#3}{}{ \mid #3} #1\}}

\newcommand{\mDom}{[0,1]^{S \times S}}
\newcommand{\dist}[1][\lambda]{\delta_{#1}}
\newcommand{\discr}[2][\lambda]{\gamma_{#1}^{\mathcal{#2}}}
\newcommand{\ordC}[1][\lambda]{\trianglelefteq_{#1}}

\newcommand{\coupling}[2]{\Omega(#1,#2)}

\theoremstyle{plain}

\begin{document}

\title[On-the-Fly Computation of Bisimilarity Distances]{On-the-Fly Computation of Bisimilarity Distances\rsuper*} 

\author[G.\ Bacci]{Giorgio Bacci}	
\address{Dept.\ of Computer Science, Aalborg University, Denmark}	
\email{\{grbacci, giovbacci, kgl, mardare\}@cs.aau.dk}  
\thanks{Work supported by 
Sapere Aude: DFF-Young
Researchers Grant 10-085054 of the Danish Council for Independent Research, 
by the VKR Center of Excellence MT-LAB and by the Sino-Danish Basic Research Center IDEA4CPS.
}	

\author[G.\ Bacci]{Giovanni Bacci}	
\address{\vspace{-18 pt}}	

\author[K.\ G.\ Larsen]{Kim G. Larsen}	
\address{\vspace{-18 pt}}	

\author[R.\ Mardare]{Radu Mardare}	
\address{\vspace{-18 pt}}	



\titlecomment{{\lsuper*}An earlier version of this paper appeared as~\cite{BacciLM:tacas13}. The present version extends~\cite{BacciLM:tacas13} by considering the case of CTMCs and improves the linear programming approach of~\cite{ChenBW12}.}


\begin{abstract}
We propose a distance between continuous-time Markov chains (CTMCs) and study the problem of computing it by comparing three different algorithmic methodologies: iterative, linear program, and on-the-fly.

In a work presented at FoSSaCS'12, Chen et al.\ characterized the bisimilarity distance of Desharnais et al.\ between discrete-time Markov chains as an optimal solution of a linear program that can be solved by using the ellipsoid method.
Inspired by their result, we propose a novel linear program characterization to compute the distance in the continuous-time setting. Differently from previous proposals, ours has a number of constraints that is bounded by a polynomial in the size of the CTMC. This, in particular, proves that the distance we propose can be computed in polynomial time.

Despite its theoretical importance, the proposed linear program characterization turns out to be inefficient in practice. Nevertheless, driven by the encouraging results of our previous work presented at TACAS'13, we propose an efficient on-the-fly algorithm, which, unlike the other mentioned solutions, computes the distances between two given states avoiding an exhaustive exploration of the state space. 
This technique works by successively refining over-approximations of the target distances using a greedy strategy, which ensures that the state space is further explored only when the current approximations are improved.

Tests performed on a consistent set of (pseudo)randomly generated CTMCs show that our algorithm improves, on average, the efficiency of the corresponding iterative and linear program methods with orders of magnitude.
\end{abstract}

\maketitle

\section*{Introduction} \label{sec:intro}

Continuous-time Markov chains (CTMCs) are one of the most prominent models in performance and dependability analysis. They are exploited in a broad range of applications, and constitute the underlying semantics of many modeling formalisms for real-time probabilistic systems such as Markovian queuing networks, stochastic process algebras, and calculi for systems biology.
An example of CTMC is presented in Figure~\ref{fig:ctmcwithepsilon}(left). Here, state $s_1$ goes to state $s_3$ and $s_4$ with probability $\frac{1}{3}$ and $\frac{2}{3}$, respectively. Each state has an associated \emph{exit-rate} representing the rate of an exponentially distributed random variable that characterizes the residence-time in the state%
\footnote{Note that the only residence time distributions that ensure the Markov property (a.k.a., memoryless transition probability) are exponential distributions.}. For example, the probability to move from $s_1$ to any other state within time $t \geq 0$ is given by $\int_0^t 3e^{-3x} dx = 1 - e^{-3t}$. A state with no outgoing transitions (as $s_3$ in Figure~\ref{fig:ctmcwithepsilon}) is called \emph{absorbing}, and represents a terminating state of the system.

\begin{figure}[b]
\centering
\begin{tikzpicture}[
  label1/.style={circle, draw=red!50,fill=red!20, thin, circle split, inner sep=0.6mm},
  label2/.style={circle,draw=blue!50,fill=blue!20, thin, circle split, inner sep=0.6mm},
  label3/.style={circle,draw=green!50!black,fill=green!80!blue,thin}]
  
\draw 
  (7.5,-2) node[label1] (s1) {$s_2$ \nodepart{lower} $3$}
  ($(s1) +(left:5)$) node[label1] (s2) {$s_1$  \nodepart{lower} $3$}
  ($(s1)!0.5!(s2)$) node[label3] (s3) {$s_3$}
  ($(s2)!0.3!(s1)+(down:2)$) node[label2] (s4) {$s_4$  \nodepart{lower} $5$}
  ($(s1)!0.3!(s2)+(down:2)$) node[label2] (s5) {$s_5$  \nodepart{lower} $5$}
  ;

\path[-latex, font=\small]
  (s1) edge[bend right=20]  node[pos=0.3,below] {$\frac{1}{3}$} (s4) %
  (s1) edge node[above] {$\frac{1}{3}$} (s3) %
  (s1) edge[bend left] node[right] {$\frac{1}{3}$} (s5) %
  
  (s2) edge node[above] {$\frac{1}{3}$} (s3) %
  (s2) edge[bend right] node[left] {$\frac{2}{3}$} (s4) %
  
  
  (s4) edge[bend right] node[below] {$1$} (s5) %
  (s5) edge[bend right] node[above] {$1$} (s4)
  ;

\end{tikzpicture}
\qquad\qquad
\begin{tikzpicture}[
  label1/.style={circle, draw=red!50,fill=red!20, thin, circle split, inner sep=0.6mm},
  label2/.style={circle,draw=blue!50,fill=blue!20, thin, circle split, inner sep=0.6mm},
  label3/.style={circle,draw=green!50!black,fill=green!80!blue,thin}]
  
\draw 
  (7.5,-2) node[label1] (s1) {$t_2$ \nodepart{lower} $3$}
  ($(s1) +(left:5)$) node[label1] (s2) {$t_1$  \nodepart{lower} $3$}
  ($(s1)!0.5!(s2)$) node[label3] (s3) {$t_3$}
  ($(s2)!0.3!(s1)+(down:2)$) node[label2] (s4) {$t_4$  \nodepart{lower} $5$}
  ($(s1)!0.3!(s2)+(down:2)$) node[label2] (s5) {$t_5$  \nodepart{lower} $5$}
  ;

\path[-latex, font=\small]
  (s1) edge[bend right=20]  node[pos=0.3,below] {$\frac{1}{3}$} (s4) %
  (s1) edge node[above] {$\frac{1}{3}$} (s3) %
  (s1) edge[bend left] node[right] {$\frac{1}{3}$} (s5) %
  
  (s2) edge node[above] {$\frac{1}{3} + \varepsilon$} (s3) %
  (s2) edge[bend right] node[left] {$\frac{2}{3} - \varepsilon$} (s4) %
  
  
  (s4) edge[bend right] node[below] {$1$} (s5) %
  (s5) edge[bend right] node[above] {$1$} (s4)
  ;

\end{tikzpicture}
\caption{A CTMC (left) and an $\varepsilon$-perturbation of it (right), for some $\varepsilon \in (0, \frac{2}{3})$. Labels are represented by different colors; states are additionally labelled with their exit rates; and transitions with probability $0$ are omitted.}
\label{fig:ctmcwithepsilon}
\end{figure}

A key concept for reasoning about the equivalence of probabilistic systems is Larsen and Skou's  \emph{probabilistic bisimulation} for discrete-time Markov chains (MCs). This notion has been extended to several types of probabilistic systems, including CTMCs. In Figure~\ref{fig:ctmcwithepsilon}(left) $s_4$ and $s_5$ are bisimilar. Moreover, although $s_1$ and $s_2$ move with different probabilities to states $s_4$ and $s_5$, their probabilities to reach any bisimilarity class is the same, so that, also $s_1$ and $s_2$ are bisimilar. 

However, when the numerical values of probabilities are based on statistical sampling or subject to error estimates, any behavioral analysis based on a notion of equivalence is too fragile, as it only relates processes with identical behaviors. 
This issue is illustrated in Figure~\ref{fig:ctmcwithepsilon}(right), where the states $t_1$ and $t_2$ (i.e., the counterparts of $s_1$ and $s_2$, respectively, after a perturbation of the transition probabilities) are not bisimilar. A similar situation occurs considering perturbations on the exit-rates or on associated labels, if one assumes they are taken from a metric space.

This is a common issue in applications such as, systems biology \cite{ThorsleyK10}, planning \cite{ComaniciP11}, games \cite{ChatterjeeAMR10}, or security \cite{CaiG09}, where one is interested in knowing whether two processes that may differ by a small amount in the real-valued parameters (probabilities, rates, etc.) have ``sufficiently'' similar behaviors. This motivated the development of the metric theory for probabilistic systems, initiated by Desharnais et al.~\cite{DesharnaisGJP04} and greatly developed and explored by De Alfaro, van Breugel, Worrell, and others \cite{AlfaroMRS07,BreugelW06,BreugelSW08}. It consists in proposing a \emph{bisimilarity distance (pseudometric)}, which measures the behavioral similarity of two models. 
These pseudometrics, e.g., the one proposed by Desharnais et al., are parametric in a \emph{discount factor} that controls the significance of the future in the measurements.

Since van Breugel et al.\ have presented a fixed point characterization of the aforementioned pseudometric in~\cite{BreugelW:icalp01}, several iterative algorithms have been developed in order to compute its approximation up to any degree of accuracy~\cite{FernsPP04,BreugelW06,BreugelSW08}.
Recently, Chen et al.~\cite{ChenBW12} proved that, for finite MCs with rational transition function, the bisimilarity pseudometrics can be computed exactly in polynomial time. The proof consists in describing the pseudometric as the solution of a linear program that can be solved using the \emph{ellipsoid method}. Although the ellipsoid method is theoretically efficient, \emph{``computational experiments with the method are very discouraging and it is in practice by no means a competitor of the, theoretically inefficient, simplex method''}, as stated in \cite{Schrijver86}. Unfortunately, in this case the simplex method cannot be used to speed up performances in practice, since the linear program to be solved may have an exponential number of constraints in the number of states of the MC.

In this paper, we introduce a bisimilarity pseudometric over CTMCs that extends that of Desharnais et al.\ over MCs, and we consider the problem of computing it both from a theoretical and practical point of view.
We show that the proposed distance can be computed in polynomial time in the size of the CTMC. This is obtained by reducing the problem of computing the distance to that of finding an optimal solution of a linear program that can be solved using the ellipsoid method. Notably, differently from the proposal in~\cite{ChenBW12}, our linear program characterization has a number of constraints that is bounded by a polynomial in the size of the CTMC. This, in particular, allows one to avoid the use of the ellipsoid algorithm in favor of the simplex or the interior-point methods.

However, also in this case, the linear program characterization turns out to be inefficient in practice, even for small CTMCs. Nevertheless, supported by the encouraging results in our previous work~\cite{BacciLM:tacas13}, we propose to follow an on-the-fly approach for computing the distance. 
This is inspired by an alternative characterization of the bisimilarity pseudometric based on the notion of \emph{coupling structure} for a CTMC. Each coupling structure is associated with a \emph{discrepancy function} that represents an over-approximation of the distance. The problem of computing the pseudometric is then reduced to that of searching for an \emph{optimal} coupling structure whose associated discrepancy coincides with the distance.
The exploration of the coupling structures is based on a greedy strategy that, given a coupling structure, moves to a new one by ensuring an actual improvement of the current discrepancy function. This strategy will eventually find an optimal coupling structure. The method is sound independently from the initial starting coupling structure.
Notably, the moving strategy is based on a \emph{local} update of the current coupling structure. Since the update is local, when the goal is to compute the distance only between certain pairs of states, the construction of the coupling structures can be done \emph{on-the-fly}, delimiting the exploration only on those states that are demanded during the computation.

The efficiency of our algorithm has been evaluated on a significant set of randomly generated CTMCs. The results show that our algorithm performs orders of magnitude better than the corresponding iterative and linear program implementations. Moreover, we provide empirical evidence that our algorithm enjoys good execution running times.

One of the main practical advantages of our approach consists in that one can focus on computing only the distances between states that are of particular interest. This is useful in practice, for instance when large systems are considered and visiting the entire state space is computationally expensive. A similar issue has been considered by Comanici et al., in~\cite{ComaniciPP12} in the case of Markov decision processes with rewards, who noticed that for computing the approximated pseudometric one does not need to update the current value for all the pairs at each iteration, but it is sufficient only to focus on the pairs where changes are happening rapidly. In our approach, the termination condition is checked locally, still ensuring that the local optimum corresponds to the global one. 
Our methods can also be used in combination with approximation techniques as, for instance, to provide a least over-approximation of the behavioral distance given over-estimates of some particular distances.

\paragraph{\textit{Synopsis:}} 
The paper is organized as follows. In Section~\ref{sec:CTMC}, we recall the basic preliminaries on continuous-time Markov chains and define the bisimilarity pseudometric. Section~\ref{sec:lp} is devoted to the analysis of the complexity of the problem of computing such a distance. Here, two approaches are considered: an approximate iterative method and a linear program characterization. In Section~\ref{sec:couplchar}, we provide an alternative characterization of the distance based on the notion of coupling structure. This is the basis for the development of an on-the-fly algorithm (Section~\ref{sec:onthefly}) for the computation of the pseudometric, whose correctness and termination is proven in Section~\ref{sec:algorithm}. The efficiency of this algorithm is supported by experimental results, shown in Section~\ref{sec:experiments}. Final remarks and conclusions are in Section~\ref{sec:conclusions}.

\section{Continuous-time Markov Chains and Bisimilarity Pseudometrics} \label{sec:CTMC}

We recall the definitions of (finite) $L$-labelled continuous-time Markov chains (CTMCs) for a nonempty set of labels $L$, and \emph{stochastic bisimilarity} over them. Then, we introduce a \emph{behavioral pseudometric} over CTMCs to be considered as a quantitative generalization of stochastic bisimilarity.

Given a finite set $X$, a discrete probability distribution over it is a function $\mu \colon X \to [0,1]$ such that $\mu(X) =1$, where $\mu(E) = \sum_{x \in E} \mu(x)$, for $E \subseteq X$. We denote the set of finitely supported discrete probability distributions over $X$ by $\Dist(X)$.

\begin{defi}[{Continuous-time Markov chain}]
\label{def:CTMC}
An $L$-labelled \emph{continuous-time Markov chain} is a tuple $\M = (S, A, \tau, \rho, \ell)$ consisting of a nonempty finite set $S$ of \emph{states}, a set $A \subseteq S$ of \emph{absorbing states}, 
a \emph{transition probability function} $\tau \colon S \setminus A \to \Dist(S)$, an \emph{exit rate function} $\rho \colon S \setminus A \to \R[>0]$, and a \emph{labeling function} $\ell \colon S \to L$.
\qee
\end{defi}

The labels in $L$ represent properties of interest that hold in a particular state according to the labeling function $\ell \colon S \to L$.  If $s \in S$ is the current state of the system and $E \subseteq S$ is a subset of states, $\tau(s)(E) \in [0,1]$ corresponds to the probability that a transition from $s$ to arbitrary $s' \in E$ is taken, and $\rho(s) \in \R[>0]$ represents the rate of an exponentially distributed random variable that characterizes the residence time in the state $s$ before any transition is taken. Therefore, the probability to make a transition from state $s$ to any $s' \in E$ within time unit $t \in \R$ is given by $\tau(s)(E) \cdot \Exp[\rho(s)]([0,t))$, where $\Exp[r](B) = \int_B r e^{-rx} \;dx$, for any Borel subset $B \subseteq \R$ and $r > 0$.
Absorbing states in $A \subseteq S$ are used to represent termination or deadlock states. An example of a CTMC is shown in Figure~\ref{fig:ctmcwithepsilon}.

For discrete-time Markov chains, the standard notion of behavioral equivalence is probabilistic bisimulation of Larsen and Skou~\cite{LarsenS91}. The following definition extends it to CTMCs. To ease the notation, for $E \subseteq S$, we introduce the relation ${\equiv_E} \subseteq S \times S$ defined by $s \equiv_E s'$ if either $s,s' \in E$ or $s,s' \notin E$.
\begin{defi}[Stochastic Bisimulation] \label{def:bisim}
Let $\M = (S, A, \tau, \rho, \ell)$ be a CTMC. An equivalence relation $R \subseteq S \times S$ is a \emph{stochastic bisimulation} on $\M$ if whenever $s \mathrel{R} t$, then
\begin{enumerate}[label=(\roman*)]
  \item $s \equiv_A t$, $\ell(s) = \ell(t)$, and
  \item if $s,t \not\in A$, then $\rho(s) = \rho(t)$ and, for all $C \in S/_R$, $\tau(s)(C) = \tau(t)(C)$.
\end{enumerate}
Two states $s, t \in S$ are \emph{bisimilar} with respect to $\M$, written $s \sim_\M t$, if they are related by some probabilistic bisimulation on $\M$.
\end{defi}
Intuitively, two states are bisimilar if they have the same labels, they agree on being absorbing or not, and, in the case they are non-absorbing, their residence-time distributions and probability of moving 
by a single transition to any given class of bisimilar states is always the same. 
As an example of two stochastic bisimilar states, consider $s_1$ and $s_2$ in the CTMC depicted on the left hand side of Figure~\ref{fig:ctmcwithepsilon}. A bisimulation relation that relates them is the equivalence relation with equivalence classes given by $\set{s_1,s_2}{}$, $\set{s_3}{}$, and $\set{s_4, s_5}{}$.

\subsection{Bisimilarity Pseudometrics on CTMCs} 
In this section, we introduce a family of pseudometrics on CTMCs parametric in a discount factor $\lambda \in (0,1)$. Following the approach of~\cite{BreugelHMW07}, given a CTMC $\M = (S, A, \tau, \rho, \ell)$ we define a ($1$-bounded) pseudometric on $S$ as the least fixed point of an operator on the set $\mDom$ of functions from $S \times S$ to $[0,1]$. This pseudometric is then shown to be adequate with respect to stochastic bisimilarity: we prove that two states are stochastic bisimilar if and only if they have distance zero.

Recall that $d \colon X \times X \to \R$ is a \emph{pseudometric} on a set $X$ if $d(x,x) = 0$ (reflexivity), $d(x,y) = d(y,x)$ (symmetry)  and $d(x,y) + d(y,z) \geq d(x,z)$ (triangular inequality), for arbitrary $x,y,z \in X$; it is a \emph{metric} if, in addition, $d(x,y) = 0$ iff $x = y$. A pair $(X,d)$ where $d$ is a (pseudo)metric on $X$ is called a \emph{(pseudo)metric space}.

\noindent Hereafter we will assume $L$ to be equipped with a $1$-bounded metric%
\footnote{Since the set $S$ of states is assumed to be finite, one may assume the set of labels to be so as well. Thus, the metric $d_L$ on labels can be bounded without loss of generality.} $d_L \colon L \times L \to [0,1]$.

The operator we are going to introduce will use three key ingredients: the distance $d_L$ between labels, a distance between residence-time distributions, and a distance between transition distributions. The first is meant to measure the \emph{static} differences with respect to the labels associated with the states, the last two are meant to capture the differences in the \emph{dynamics}, respectively, with respect to the continuous and discrete probabilistic choices. 

To this end, we consider two distances over probability distributions. The first one is the \emph{total variation metric}, defined for arbitrary Borel probability measures $\mu, \nu$ over $\R$ as  
\begin{equation*}
  \textstyle
  \tv{\mu}{\nu} = \sup_{E} |\mu(E) - \nu(E)| \,,
\end{equation*}
where the supremum is taken over the Borel measurable sets of $\R$.
The second one is the \emph{Kantorovich distance}, which is based on the notion of \emph{coupling} of probability measures, which we introduce next in the case of probability distributions over finite sets.
\begin{defi}[Coupling]
Let $S$ be a finite set, and let $\mu, \nu \in \Dist(S)$. A probability distribution $\omega \in \Dist(S \times S)$ is
said a \emph{coupling} for $(\mu, \nu)$ if, for arbitrary $u,v \in S$ 
\begin{align*}
  \textstyle \sum_{v \in S} \omega(u, v) = \mu(u) 
  && \text{and} && 
  \textstyle \sum_{u \in S} \omega(u, v) = \nu(v) \,.
\end{align*}
In other words, $\omega$ is a joint probability distribution with left and right marginal, respectively, given by $\mu$ and $\nu$. We denote the set of couplings for $(\mu, \nu)$ by $\coupling{\mu}{\nu}$.
\end{defi}

For a finite set $S$ and a $1$-bounded distance $d \colon S \times S \to [0,1]$ over it, the \emph{Kantorovich distance} is defined, for arbitrary distributions $\mu, \nu \in \Dist(S)$ as follows%
\footnote{The minimum can be used in place of an infimum thanks to the Fenchel-Rockafeller duality theorem (see~\cite[Theorems 1.3 and 1.9]{Villani03}).}
\begin{equation*}
  \textstyle
  \mathcal{K}_d(\mu, \nu) = \min \set{ \sum_{u,v \in S} d(u,v) \cdot \omega(u,v) }{\omega \in \coupling{\mu}{\nu}} \,.
\end{equation*}
Intuitively, $\mathcal{K}_d$ lifts a ($1$-bounded) distance over $S$ to a ($1$-bounded) distance over its probability distributions. One can show that $\mathcal{K}_d$ is a (pseudo)metric if $d$ is a (pseudo)metric.

Now, consider the following functional operator.
\begin{defi} \label{def:deltaoperator}
Let $\M = (S, A, \tau, \rho, \ell)$ be CTMC and $\lambda \in (0,1)$ a \emph{discount factor}.
The function $\Delta^\M_\lambda \colon \mDom \to \mDom$ is defined as follows, for $d \colon S \times S \to [0,1]$ and $s, t \in S$
\begin{equation*}
  \Delta^\M_\lambda(d)(s,t) =
  \begin{cases}
    1 & \text{if $s \not\equiv_A t$} \\
    \lbld{s}{t} & \text{if $s,t \in A$} \\
    \max \set{\lbld{s}{t}, \lambda \cdot \trd{s}{t} }{} & \text{if $s,t \notin A$} \\
  \end{cases}
\end{equation*}
where $\trd[]{}{} \colon \mDom \to \mDom$ and $\lbld{}{}, \rated{}{} \colon S \times S \to [0,1]$ are respectively defined by
\begin{gather*}
  \trd{s}{t} = \rated{s}{t}  + (1 - \rated{s}{t}) \cdot \mathcal{K}_d(\tau(s),\tau(t)) \,,  \\
  \lbld{s}{t} = d_L(\ell(s),\ell(t)) \,,
  \qquad \text{and} \qquad
  \rated{s}{t} = \tv{\Exp[\rho(s)]}{\Exp[\rho(t)]}  \,.
\end{gather*}
\end{defi}
The functional $\Delta^\M_\lambda$ measures the difference of two states with respect to: their labels (by means of the pseudometric $\lbld{}{}$), their residence-time distributions (by means of the pseudometric $\rated{}{}$), and their discrete probabilities to move to the next state (by means of the Kantorovich distance). If two states disagree on being absorbing (or not) they are considered incomparable, and their distance is set to $1$. If both states are absorbing, they express no dynamic behavior, hence they are compared statically, and their distance corresponds to that occurring between their respective labels. Finally, if the states are non-absorbing, then they are compared with respect to both their static and dynamic features, namely, taking the maximum among their respective associated distances.

Specifically, the value $\rated{s}{t}$ corresponds to the least probability that two transitions are taken independently from the states $s$ and $t$ at different moments in time. 
This value is used by the functional $\trd[]{}{}$ to measure the overall differences that might occur in the dynamics of the two states in combination with the Kantorovich distance between their transition probability distributions (see Remark~\ref{rmk:altenativedefoftimedistance} for more details).

The set $\mDom$ is endowed with the partial order $\sqsubseteq$ defined as $d \sqsubseteq d'$ iff $d(s,t) \leq d'(s,t)$ for all $s,t \in S$ and it forms a complete lattice. The bottom element ${\bf 0}$ is the constant $0$ function, while the top element is the constant $1$ function. For any subset $D \subseteq \mDom$, the least upper bound $\bigsqcup D$, and greatest lower bound $\bigsqcap D$ are, respectively, given by $(\bigsqcup D)(s,t) = \sup_{d \in D} d(s,t)$ and $(\bigsqcap D)(s,t) = \inf_{d \in D} d(s,t)$, for all $s,t \in S$.

It is easy to check that, for any $\M$ and $\lambda \in (0,1)$, $\Delta^\M_\lambda$ is monotone (i.e., whenever $d \sqsubseteq d'$, then $\Delta^\M_\lambda(d) \sqsubseteq \Delta^\M_\lambda(d')$), thus, since $(\mDom, \sqsubseteq)$ is a complete lattice, by Tarski's fixed point theorem $\Delta^\M_\lambda$ admits least and greatest fixed points.

\begin{defi}[Bisimilarity distance] \label{def:Delta-fix}
Let $\M$ be a CTMC and $\lambda \in (0,1)$.
The $\lambda$-disc\-oun\-ted \emph{bisimilarity pseudometric} on $\mathcal{M}$,  
denoted by $\dist^\M$, is the least fixed point of $\Delta^{\mathcal{M}}_{\lambda}$.
\end{defi}

The rest of the section is devoted to show that the least fixed point $\dist^\M$ is indeed a pseudometric and, moreover, is adequate with respect to stochastic bisimilarity (Theorem~\ref{th:bisimpseudo}). This justifies the definition above.
To this end we need some technical lemmas. In particular, we prove that $\Delta^\M_\lambda$ preserves pseudometrics (Lemma~\ref{lem:preservepseudometric}) and it is Lipschitz continuous (Lemma~\ref{lem:nonexpansive}). 

Hereafter, unless mentioned otherwise, we fix a CTMC $\M = (S, A, \tau, \rho, \ell)$ and a discount factor $\lambda \in (0,1)$. To ease the notation $\Delta^\M_{\lambda}$, $\dist^\M$, and $\sim_\M$ will be denoted simply by $\Delta_{\lambda}$, $\dist$, and $\sim$, respectively, whenever $\M$ is clear from the context.

\begin{lem} \label{lem:preservepseudometric}
The operator $\Delta_\lambda$ preserves pseudometrics.
\end{lem}
\begin{proof}
Let $d \colon S \times S \to [0,1]$ be a pseudometric. We want to prove that $\Delta_\lambda(d)$ is a pseudometric. Recall that $\lbld{}{}, \rated{}{} \colon S \times S \to [0,1]$ are pseudometrics. Thus, since the point wise maximum of pseudometrics is a pseudometric, it suffices to prove that $\trd[]{}{}$ preserves pseudometrics. Recall that, $\mathcal{K}_d \colon \Dist(S) \times \Dist(S) \to [0,1]$ is a pseudometric, since $d$ is so. Thus, reflexivity and symmetry are immediate. The only nontrivial case is the triangular inequality. Let $s,t,u \in S \setminus A$, we want to prove $\trd{s}{t} \leq \trd{s}{u} + \trd{u}{t}$. First note that, for $0 \leq \beta \leq 1$ and $\alpha' \geq \alpha$, the following holds:
\begin{align*}
  \alpha + (1- \alpha) \beta
  &
  = \beta + (1 - \beta) \alpha \tag{distributivity} \\
  &\leq \beta + (1 - \beta) \alpha' \tag{$\alpha \leq \alpha'$ and $0 \leq \beta \leq 1$} \\
  &= \alpha' + (1 - \alpha') \beta  \,. \tag{distributivity}
\end{align*}
Thus, since $\rated{}{}$ is a pseudometric, by triangular inequality and the above we have 
\begin{align*}
  \trd{s}{t} 
  &= \rated{s}{t}  + (1 - \rated{s}{t}) \cdot \mathcal{K}_d(\tau(s),\tau(t))  \tag{def. $\trd[]{}{}$} \\
  &\leq \rated{s}{u} + \rated{u}{t} + \big(1 - (\rated{s}{u} + \rated{u}{t}) \big) \cdot \mathcal{K}_d(\tau(s),\tau(t))
  \label{eq:ineq} \tag{$*$}
\intertext{%
If we show that the last summand in \eqref{eq:ineq} is less than or equal to the sum of $(1 - \rated{s}{u}) \cdot \mathcal{K}_d(\tau(s),\tau(u))$ and $(1 - \rated{u}{t}) \cdot \mathcal{K}_d(\tau(u),\tau(t))$,
we get the following, and we are done
}
&\leq \rated{s}{u}  + (1 - \rated{s}{u}) \cdot \mathcal{K}_d(\tau(s),\tau(u)) +
   \rated{u}{t}  + (1 - \rated{u}{t}) \cdot \mathcal{K}_d(\tau(u),\tau(t)) \\
&= \trd{s}{u} +   \trd{u}{t}  \tag{def. $\trd[]{}{}$} \,.
\end{align*}
To this end, consider two cases. If $\rated{s}{u} + \rated{u}{t} > 1$ then the inequality holds trivially, since $1 - (\rated{s}{u} + \rated{u}{t}) < 0$, so that the last summand in \eqref{eq:ineq} is negative and $1-\rated{s}{u} \geq 0$. 
Instead, if $\rated{s}{u} + \rated{u}{t} \leq 1$ then $1 - (\rated{s}{u} + \rated{u}{t}) \geq 0$, so we have
\begin{align*}
  &\big(1 - (\rated{s}{u} + \rated{u}{t}) \big) \cdot \mathcal{K}_d(\tau(s),\tau(t)) \\
  &\quad\leq \big(1 - (\rated{s}{u} + \rated{u}{t}) \big) \cdot 
    \big(\mathcal{K}_d(\tau(s),\tau(u)) + \mathcal{K}_d(\tau(u),\tau(t)) \big) 
  \tag{triang. $\mathcal{K}_d$} \\
  &\quad= 
  \big(1 - (\rated{s}{u} + \rated{u}{t}) \big) \cdot \mathcal{K}_d(\tau(s),\tau(u)) +
  \big(1 - (\rated{s}{u} + \rated{u}{t}) \big) \cdot  \mathcal{K}_d(\tau(u),\tau(t)) \\
  &\quad\leq 
  \big(1 - \rated{s}{u} \big) \cdot \mathcal{K}_d(\tau(s),\tau(u)) +
  \big(1 - \rated{u}{t} \big) \cdot  \mathcal{K}_d(\tau(u),\tau(t))
\end{align*}
and we are done.
\end{proof}

The set $\mDom$ can be turned into a metric space by means of the supremum norm $\norm{d}{d'} = \sup_{s,t \in S} |d(s,t) - d'(s,t)|$. Next we show that the $\lambda$-discounted functional operator $\Delta_\lambda$ is $\lambda$-Lipschitz continuous, that is $\norm{\Delta_\lambda(d')}{\Delta_\lambda(d)} \leq \lambda \cdot \norm{d'}{d}$, for any $d,d' \in \mDom$.
\begin{lem} \label{lem:nonexpansive}
The operator $\Delta_\lambda$ is $\lambda$-Lipschitz continuous.
\end{lem}
\proof
By~\cite[Corollary 1]{Breugel12}, to prove that $\Delta_\lambda$ is $\lambda$-Lipschitz continuous it suffices to show that, whenever $d \sqsubseteq d'$ then, for all $s,t \in S$, $\Delta_\lambda(d')(s, t) - \Delta_\lambda(d)(s, t) \leq \lambda \cdot \norm{d'}{d}$.
If $s \not\equiv_A t$, then $\Delta_\lambda(d')(s, t) - \Delta_\lambda(d)(s, t) = 1 - 1 = 0$, so that, the inequality is satisfied. If $s,t \notin A$, $\Delta_\lambda(d')(s, t) - \Delta_\lambda(d)(s, t) = \lbld{s}{t} - \lbld{s}{t} = 0$, and again the inequality holds trivially. The same happens when $\lbld{s}{t} \geq \lambda \cdot \trd[d']{s}{t}$. Indeed, by monotonicity of $\trd[]{}{}$, $\lbld{s}{t} \geq \lambda \cdot \trd{s}{t}$ holds, so that, by definition of $\Delta_\lambda$ we have
\begin{equation*}
\Delta_\lambda(d')(s,t) - \Delta_\lambda(d)(s,t)
  = \lambda \cdot \big(  \lbld{s}{t} - \lbld{s}{t} \big) = 0 \leq \lambda \cdot \norm{d'}{d} \,. 
\end{equation*}
It remains the case when $s,t \notin A$ and $\lambda \cdot \trd{s}{t} \leq \lbld{s}{t} < \lambda \cdot \trd[d']{s}{t}$. 
Assume that $\mathcal{K}_d(\tau(s),\tau(t)) = \sum_{u,v \in S} d(u,v) \cdot \omega(u,v)$, for some $\omega \in \coupling{\tau(s)}{\tau(t)}$, then we have
\begin{align*}
  \Delta_\lambda(d')(s,t) - \Delta_\lambda(d)(s,t)
  &\leq \lambda \cdot \big( \trd[d']{s}{t} - \trd[d]{s}{t} \big) \\
  &\leq \lambda \cdot \big( \mathcal{K}_{d'}(\tau(s),\tau(t)) - \mathcal{K}_d(\tau(s),\tau(t)) \big) \\
  &\leq \textstyle \lambda \cdot \big(
  \sum_{u,v \in S} d'(u,v) \cdot \omega(u,v) - \sum_{u,v \in S} d(u,v) \cdot \omega(u,v) \big) \\
  &= \textstyle \lambda \cdot \big(
  \sum_{u,v \in S} (d'(u,v) - d(u,v)) \cdot \omega(u,v) \big) \\
  &\leq \textstyle \lambda \cdot \big(
  \sum_{u,v \in S} \norm{d'}{d} \cdot \omega(u,v) \big) \\
  &= \textstyle \lambda \cdot \big(
   \norm{d'}{d} \cdot \sum_{u,v \in S} \omega(u,v) \big) \\
  &= \lambda \cdot \norm{d'}{d} \,.\rlap{\hbox to 221 pt{\hfill\qEd}}
\end{align*}

It is standard that $\mDom$ with the supremum norm forms a complete metric space (i.e, every Cauchy sequence converges). Therefore, since $\lambda \in (0,1)$, a direct consequence of Lemma~\ref{lem:nonexpansive} and Banach's fixed point theorem is the following.
\begin{thm} \label{th:uniquefix}
For any $\lambda \in (0,1)$, $\dist$ is the unique fixed point of $\Delta_\lambda$. Moreover, for any $n \in \N$, and $d \colon S \times S \to [0,1]$ we have $\norm{\dist}{\Delta^n_\lambda(d)} \leq \frac{\lambda^n}{1-\lambda} \norm{\Delta_\lambda(d)}{d}$.
\end{thm}

Now we are ready to state the main theorem of this section.
\begin{thm}[Bisimilarity pseudometric] \label{th:bisimpseudo}
$\dist$ is a pseudometric. Moreover, for any $s,t \in S$, $s \sim t$ if and only if $\dist(s,t) = 0$
\end{thm}
\begin{proof}
We first prove that $\dist$ is a pseudometric. By Lemma~\ref{lem:nonexpansive} and Banach's fixed point theorem, $\dist = \bigsqcup_{n \in \N} \Delta^n_\lambda(\mathbf{0})$. Clearly, $\mathbf{0}$ is a pseudometric.  Thus by Lemma~\ref{lem:preservepseudometric}, a simple induction on $n$ shows that, for all $n \in \N$, $\Delta^n_\lambda(\mathbf{0})$ is a pseudometric. Since the least upper bound with respect to $\sqsubseteq$ preserves pseudometrics, we have that $\dist$ is so.


Now we are left to prove that, for any $s,t \in S$, $s \sim t$ iff $\dist(s,t) = 0$.
\begin{trivlist} \topsep0ex
\item ($\Leftarrow$) We prove that $R = \set{ (s,t) }{ \dist(s,t) = 0 }$ is a stochastic bisimulation. Clearly, $R$ is an equivalence. Assume $(s,t) \in R$, then, by definition of $\Delta_\lambda$, one of the following holds:
\begin{enumerate}[label=(\roman*)]
  \item \label{itm:th:fixedbisimdist1} $s,t \in A$ and $\lbld{s}{t} = 0$; 
  \item \label{itm:th:fixedbisimdist2} $s,t \notin A$, $\lbld{s}{t} = 0$, and $\trd[\dist]{s}{t} = 0$.
\end{enumerate}
If \eqref{itm:th:fixedbisimdist1} holds, by $\lbld{s}{t} = 0$ we get that $\ell(s) = \ell(t)$.
If \eqref{itm:th:fixedbisimdist2} holds, we have $\rated{s}{t} = 0$ and $\mathcal{K}_{\dist}(\tau(s),\tau(t)) = 0$. By $\rated{s}{t} = 0$ we get $\Exp[\rho(s)] = \Exp[\rho(t)]$ and hence $\rho(s) = \rho(t)$. By \cite[Lemma 3.1]{FernsPP04}, $\mathcal{K}_{\dist}(\tau(s),\tau(t)) = 0$ implies that, for all $C \in S/_R$, $\tau(s)(C) = \tau(t)(C)$. Therefore $\mathrel{R}$ is a bisimulation.

\item ($\Rightarrow$) 
Let $R \subseteq S \times S$ be a stochastic bisimulation on $\M$, and define $d_R \colon S \times S \to [0,1]$ by $d_R(s,t) = 0$ if $(s,t) \in R$ and $d_R(s,t) = 1$ if $(s,t) \notin R$. We show that $\Delta_\lambda(d_R) \sqsubseteq d_R$. If $(s,t) \notin R$, then $d_R(s,t) = 1 \geq \Delta_\lambda(d_R)(s,t)$. If 
$(s,t) \in R$, then $\ell(s) = \ell(t)$ and one of the following holds:
\begin{enumerate}[label=(\roman*)]
	\item \label{itm:th:fixedbisimdist3} $s,t \in A$;
	\item \label{itm:th:fixedbisimdist4} $s,t \notin A$, $\rho(s) = \rho(t)$ and,
	$\forall C \in S/_{R}. \, \tau(s)(C) = \tau(t)(C)$.
\end{enumerate}
If \eqref{itm:th:fixedbisimdist3} holds, $\Delta_\lambda(d_R)(s,t) = \lbld{s}{t} = 0 = d_R(s,t)$. 
If \eqref{itm:th:fixedbisimdist4} holds, by \cite[Lemma 3.1]{FernsPP04} and the fact that, for all $C \in S/_{R}$, $\tau(s)(C) = \tau(t)(C)$, we have $\mathcal{K}_{d_R}(\tau(s),\tau(t)) = 0$. Moreover $\rated{s}{t} = 0$. This gives  that $\Delta_\lambda(d_R)(s,t) = 0 = d_R(s,t)$.

Since $\sim$ is a stochastic bisimulation, $\Delta_\lambda(d_\sim) \sqsubseteq d_\sim$, so that, by Tarski's fixed point theorem, $\dist \sqsubseteq d_\sim$. By definition of $d_\sim$ and  $\dist \sqsubseteq d_\sim$, $s \sim t$ implies $\dist(s,t) = 0$.
\qedhere
\end{trivlist}
\end{proof}

\begin{rem} \label{rmk:altenativedefoftimedistance}
Observe that Theorem~\ref{th:bisimpseudo} holds for alternative definitions of the functional $\Delta_\lambda$. 
An example is given when the functional $\trd[]{}{}$ in the definition of $\Delta_\lambda$ is given as 
\begin{equation}
  \label{eq:alterdef1}
  \trd[d]{s}{t} = \max \set{\frac{1}{M} \cdot |\rho(s) - \rho(t)|, \mathcal{K}_d(\tau(s),\tau(t))}{} \,, 
\end{equation}
where $M = \max_{s \in S} \rho(s)$ is used to rescale the symmetric difference $|\rho(s) - \rho(t)|$ to a value within $[0,1]$.
Another example is obtained by replacing the maximum above by a convex combination of the two values as below, 
for some $\alpha \in (0,1)$,
\begin{equation}
  \label{eq:alterdef1}
  \trd[d]{s}{t} = \alpha \cdot \frac{1}{M} |\rho(s) - \rho(t)| + (1- \alpha) \cdot \mathcal{K}_d(\tau(s),\tau(t)) \,, 
\end{equation}

So, what does it make a proposal for $\Delta_\lambda$ preferable to another? Although Theorem~\ref{th:bisimpseudo} is an important property for a behavioral pseudometric, it does not say much about the states that have distance different from zero. In this sense, a good behavioral metric should relate the distance with a concrete problem. Our definition of $\Delta_\lambda$, for example, is motivated by a result in~\cite{BacciLM:fossacs15} that states that the total variation distance of CTMCs (more generally, on semi-Markov chains) is logically characterized as the maximal difference w.r.t.\ the likelihood for two states to satisfy the same Metric Temporal Logic (MTL) formula~\cite{Koymans90}. It turns out that when $d_L$ is the discrete metric over $L$ (i.e., $d_L(a,b) = 0$ if $a = b$, and $1$ otherwise), $\dist$ bounds from above the total variation distance. This relates $\dist$ to the \emph{probabilistic model checking problem} of MTL-formulas against CTMCs. The above alternative proposals for a distance do not enjoy this property.
\qee
\end{rem}

\section{Complexity and Linear Programming representation} \label{sec:lp}

In this section, we study the problem of computing the bisimilarity distance by considering two different approaches. The former is an iterative method that approximates $\dist$ from below (resp.\ above) successively applying the operator $\Delta_\lambda$ starting from the least (resp.\ greatest) element in $\mDom$. The latter is based on a linear program characterization of $\dist$ that is based on the Kantorovich duality~\cite{Villani03}. In contrast to an analogous proposal in~\cite{ChenBW12}, our linear program has a number of constraints that is polynomially bounded in the size of the CTMC. As a consequence, the bisimilarity distance $\dist$ can be computed in polynomial time in the size of the CTMC.

\subsection{Iterative method} \label{sec:iterativemethod}
By Theorem~\ref{th:uniquefix}, for any $\epsilon > 0$, it follows that to get $\epsilon$-close to $\dist$, it suffices to iterate the application of the fixed point operator $\lceil \log_\lambda \epsilon \rceil$ times.
\begin{prop} \label{prop:approxiter}
For any $\epsilon > 0$ and $d \colon S \times S \to [0,1]$, $\norm{\dist}{\Delta^{\lceil \log_\lambda \epsilon \rceil}_\lambda(d)} \leq \epsilon$.
\end{prop}
\begin{proof}
By Theorem~\ref{th:uniquefix} we have $\norm{\dist}{\Delta^n_\lambda(d)} \leq \frac{\lambda^n}{1-\lambda} \norm{\Delta_\lambda(d)}{d}$ and by
$\norm{\Delta_\lambda(d)}{d} \leq 1$, we have $\norm{\dist}{\Delta^n_\lambda(d)} \leq \frac{\lambda^n}{1-\lambda}$. For $n = \log_\lambda (\epsilon - \epsilon \lambda)$, we have $\epsilon = \frac{\lambda^n}{1-\lambda}$. Therefore, by Lemma~\ref{lem:nonexpansive} and ${\lceil \log_\lambda \epsilon \rceil} \geq {\log_\lambda (\epsilon - \epsilon \lambda)}$, we have $\norm{\dist}{\Delta^{\lceil \log_\lambda \epsilon \rceil}_\lambda(d)} \leq \epsilon$.
\end{proof}

By the above result, we obtain a simple method for approximating $\dist$. If the starting point is $\mathbf{0}$ we obtain an under-approximation, whereas starting from $\mathbf{1}$ we get an over-approximation. Both the approximations can be taken arbitrary close to the exact value.

%
However, as shown in the following example, the exact distance value cannot be reached in general. This holds for any discount factor.
\begin{exa}[\cite{ChenBW12}]
Consider the $\{\text{red}, \text{blue} \}$-labeled CTMC below. 
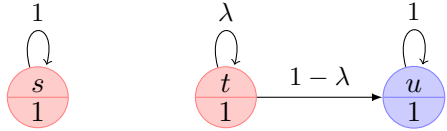
\begin{center}
\begin{tikzpicture}[
 label1/.style={circle, draw=red!50,fill=red!20, thin, circle split, inner sep=0.6mm},
 label2/.style={circle,draw=blue!50,fill=blue!20, thin, circle split, inner sep=0.6mm}]
\def\nodesep{2.5cm}
\draw 
  (0,0) node[label1] (s) {$s$ \nodepart{lower} $1$}
  ($(s) +(right:\nodesep)$) node[label1] (t) {$t$ \nodepart{lower} $1$}
  ($(t) +(right:\nodesep)$) node[label2] (u) {$u$ \nodepart{lower} $1$};

\path[-latex, font=\small]
  (s) edge[loop above] node[above] {$1$} (s)
  (t) edge[loop above] node[above] {$\lambda$} (t)
       edge node[above] {$1-\lambda$} (u)
  (u) edge[loop above] node[above] {$1$} (u);
\end{tikzpicture}
\end{center}
Let $d_L \colon L \times L \to [0,1]$ be the discrete metric over $L$, defined as $d_L(l,l') = 0$ if $l = l'$ and $1$ otherwise.
One can check that $\dist(s,t) = \frac{\lambda - \lambda^2}{1- \lambda^2}$ and, for all $n \in \N$, $\Delta^n_{\lambda}(\mathbf{0})(s,t) \leq \frac{\lambda - \lambda^{2n+1}}{1+\lambda}$. Since, for all $n \in \N$, $\frac{\lambda - \lambda^{2n+1}}{1+\lambda} < \frac{\lambda - \lambda^2}{1- \lambda^2}$, we have that the fixed point cannot be reached in a finite number of iterations. 
\qee
\end{exa}

In \cite{ChenBW12} is shown that the bisimilarity distance of Desharnais et al.~\cite{DesharnaisGJP04} can be computed \emph{exactly} by iterating the fixed point operator up to a precision that allows one to use the continued fraction algorithm to yield the exact value of the fixed point. This method can be applied provided that the pseudometric has rational values. In their case, this is ensured assuming that the transition probabilities are rational. Unfortunately, in our case this cannot be ensured under the same conditions. Indeed, the total variation distance between exponential distributions with rates $r, r' > 0$ is analytically 
solved as follows
\begin{equation}
\tv{\Exp[r]}{\Exp[r']} = 
\begin{cases}
0 &\text{if $r = r'$} \\
\displaystyle\left| \left( \frac{r'}{r} \right)^{\frac{r}{r - r'}} - 
\left( \frac{r'}{r} \right)^{\frac{r'}{r - r'}} \right| & \text{otherwise}
\end{cases}
\label{eq:tvexp}
\end{equation}
thus, even restricting to rational exit-rates and probabilities, the distance may assume irrational values. 
As a consequence, we cannot assume to compute, in general, the exact distance values.

\subsection{Linear Program Characterization} \label{sec:LPcharact}
Our linear program characterization leverages on two key results. The first one is the uniqueness of the fixed point of $\Delta_\lambda$ (Theorem~\ref{th:uniquefix}). The second one is a \emph{dual} linear program characterization of the Kantorovich distance.

For $S$ finite, $d \colon S \times S \to [0,1]$, and $\mu, \nu \in \Dist(S)$, the value $\mathcal{K}_d(\mu,\nu)$ coincides with the optimal value of the following linear program
\begin{align}
\mathcal{K}_d(\mu,\nu) &= 
\begin{aligned}[t]
\min_{\omega} & \: 
\textstyle \sum_{u,v \in S} d(u,v) \cdot \omega_{u,v} \\[-0.5ex]
	&\textstyle \sum_{v} \omega_{u,v} = \mu(u) && \forall u \in S \\
	&\textstyle \sum_{u} \omega_{u,v} = \nu(v) && \forall v \in S  \\
	& \omega_{u,v} \geq 0 && \forall u,v \in S \, .
\end{aligned}
\label{eq:primalK}
\intertext{By a standard argument in linear optimization, the above can be alternatively represented by the following dual linear program}
\mathcal{K}_d(\mu,\nu) &= 
\begin{aligned}[t]
\max_y & \textstyle\sum_{u \in S} (\mu(u) - \nu(u) )\cdot y_u \\[-1ex]
    & y_{u} - y_{v} \leq d(u,v) && \forall u,v \in S \, .
\end{aligned}
\label{eq:dualK}
\end{align}
This alternative characterization is a special case of a more general result commonly known as the \emph{Kantorovich duality} and extensively studied in linear optimization theory (see~\cite{Villani03}).

Let $n = |S|$ and $h = n - |A|$. Consider the linear program in Figure~\ref{fig:LP}, hereafter denoted by $D$, with variables $d \in \R[]^{n^2}$, $y \in \R[]^{h^2 + n}$ and  $k,m \in \R[]^{h^2}$. The constraints of $D$ are easily seen to be bounded and feasible. Moreover, the objective function of $D$ attains its optimal value when the vectors $k$ and $m$ are maximized in each component\footnote{Indeed, for arbitrary $x,y \in \R^n$ and $n \in \N$, if $x_i \leq y_i$ for all $i = 1..n$, then $\sum_{i =1}^n x_i \leq \sum_{i =1}^n y_i$.}. Therefore, according to \eqref{eq:dualK}, an optimal solution $(d^*,y^*,k^*,m^*) \in D$ satisfies the following equalities
\begin{align*}
&\forall s,t \not\in A. & \:  m^*_{s,t} &= \min \{\lbld{s}{t}, \lambda \big( \rated{s}{t} + (1 - \rated{s}{t}) \mathcal{K}_{d^*}(\tau(s), \tau(t)) \big) \} \,, \\
&\forall s,t \not\in A. \:  & k^*_{s,t} &= \mathcal{K}_{d^*}(\tau(s), \tau(t)) \,.
\end{align*}
By the above equalities and the constraints of $D$, it follows that $d^*$ is a fixed point of $\Delta_\lambda$. Since the distance is the unique fixed point of $\Delta_\lambda$, $d^* = \dist$.
\begin{figure}[t]
\fbox{
\parbox[c][5.4cm][c]{\textwidth}{
\begin{align*}
\operatorname*{arg\,max}_{d,y,k,m} & \: 
\textstyle \sum_{s,t \not\in A} k_{s,t} + m_{s,t} \\
    & d_{s,t} = 1 && \forall s, t \in S.\: s \not\equiv_A t  \\
    & d_{s,t} = \lbld{s}{t} && \forall s, t \in A  \\
    & d_{s,t} = \lbld{s}{t} + \lambda \big( \rated{s}{t} + (1 - \rated{s}{t}) k_{s,t} \big) - m_{s,t} && \forall s, t \not\in A \\
    & m_{s,t} \leq \lbld{s}{t}   && \forall s,t \not\in A \\
    & m_{s,t} \leq  \lambda \big( \rated{s}{t} + (1 - \rated{s}{t}) k_{s,t} \big)  && \forall s,t \not\in A \\
    & k_{s,t} = \textstyle \sum_{u \in S} (\tau(s)(u) - \tau(t)(u))\cdot y_u^{s,t}  && \forall s,t \not\in A\\ 
    & y_{u}^{s,t} - y_{v}^{s,t} \leq d_{u,v} && \forall s,t \not\in A, \forall u,v \in S
\end{align*}}}
\caption{Linear program characterization of $\dist$ for $\lambda \in (0,1)$ and $\M = (S, A, \tau, \rho, \ell)$.}
\label{fig:LP}
\end{figure}
\begin{thm} \label{th:LPchar}
Let $(d^*,y^*,k^*,m^*)$ be a solution of $D$. Then, for all $s,t \in S$, $d^*_{s,t} = \dist(s,t)$.
\end{thm}
\begin{proof}
Consider $d^*\in \R[]^{|S|^2}$ and the linear program above, hereafter denoted by $D(d^*)$
\begin{align*}
\operatorname*{arg\,max}_{y,k,m} & \: 
\textstyle \sum_{s,t \not\in A} k_{s,t} + m_{s,t} \\
    & m_{s,t} \leq \lbld{s}{t}   && \forall s,t \not\in A \\
    & m_{s,t} \leq  \lambda \big( \rated{s}{t} + (1 - \rated{s}{t}) k_{s,t} \big)  && \forall s,t \not\in A \\
    & k_{s,t} = \textstyle \sum_{u \in S} (\tau(s)(u) - \tau(t)(u))\cdot y_u^{s,t}  && \forall s,t \not\in A\\ 
    & y_{u}^{s,t} - y_{v}^{s,t} \leq d^*_{u,v} && \forall s,t \not\in A, \forall u,v \in S
\end{align*}
Any feasible solution of $D(d^*)$ can be improved by increasing any value of $k_{s,t}$ or $m_{s,t}$ for some $s,t \notin A$. On the one hand, the value of $k_{s,t}$ (for some $s,t \notin A$) can be increased independently from that of $(m_{s,t})_{s,t \notin A}$ since it only depends on the values of $(y^{s,t}_u)_{u \in S}$. For this reason, if we denote by $k^*_{s,t}$ an optimal value for $k_{s,t}$ we have the following equality
\begin{equation}
\begin{aligned}
     k^*_{s,t} = \max_{y} & \textstyle \sum_{u \in S} (\tau(s)(u) - \tau(t)(u) )\cdot y_u\\[-1ex]
     & y_{u} - y_{v} \leq d^*_{u,v} && \forall u,v \in S
\end{aligned} \label{eq:localKdual}
\end{equation}
Thus, by Equation~\eqref{eq:dualK}, we have that $k^*_{s,t} = \mathcal{K}_{d^*}(\tau(s), \tau(t))$ for all $s,t \not\in A$. 
On the other hand, the value of $m_{s,t}$ (for some $s,t \notin A$) is bounded from above by the constant $\lbld{s}{t}$ and the $\lambda \big( \rated{s}{t} + (1 - \rated{s}{t}) k_{s,t} \big)$. The value of $\lambda \big( \rated{s}{t} + (1 - \rated{s}{t}) k_{s,t} \big)$ increases with that of $k_{s,t}$, hence, if if we denote by $m^*_{s,t}$ the optimal value for $m_{s,t}$, we have that
\begin{align}
&\forall s,t \not\in A. & \:  m^*_{s,t} &= \min \big\{\lbld{s}{t},\, \lambda \big( \rated{s}{t} + (1 - \rated{s}{t}) k^*_{s,t} \big) \big\} \label{eq:m}\\
&\forall s,t \not\in A. \:  & k^*_{s,t} &= \mathcal{K}_{d^*}(\tau(s), \tau(t)) \, . \label{eq:k}
\end{align}
Let $(d',y',k',m')$ be an optimal solution of $D$ and $(y'',k'',m'')$ be an optimal solution of $D(d')$. 
Notice that the optimal value $\sum_{s,t \not\in A} k''_{s,t} + m''_{s,t}$ of $D(d')$ (i.e., $\sum_{s,t \not\in A} k''_{s,t} + m''_{s,t}$) is greater than or equal to the optimal value of $D$ (i.e., $\sum_{s,t \not\in A} k'_{s,t} + m'_{s,t}$), since the constraints of $D(d')$ are a subset of those of $D$.
Consider now two possible cases. \\
(Case 1:  $k' = k''$ and $m' = m''$) For $s,t \notin A$ we have that the following equalities hold
\begin{align*}
	d'_{s,t} &= \lbld{s}{t} + \lambda \big( \rated{s}{t} + (1 - \rated{s}{t}) \; k'_{s,t} \big) - m'_{s,t}  \tag{by $(d',y',k',m')$ feasible for $D$} \\
	&= \lbld{s}{t} + \lambda \big( \rated{s}{t} + (1 - \rated{s}{t}) \; k''_{s,t} \big) - m''_{s,t} \tag{by $k' = k''$ and $m' = m''$} \\
	&= \max \{ \lbld{s}{t} , \,  \lambda \big( \rated{s}{t} + (1 - \rated{s}{t}) \; k''_{s,t} \big) \}  \tag{by \eqref{eq:m}} \\
	&= \max \{ \lbld{s}{t}, \lambda \big( \rated{s}{t} + (1 - \rated{s}{t}) \; \mathcal{K}_{d'}(\tau(s),\tau(t)) \big) \} \tag{by \eqref{eq:k}} \\
	&= \max \{ \lbld{s}{t}, \lambda \mathcal{T}(d')(s,t) \} \tag{by def.\ $\mathcal{T}$} \\	
	&= \Delta(d')(s,t) \tag{by def.\ $\Delta_\lambda$}
\end{align*}
From the above equality and feasibility of $(d',y',k',m')$ for $D$, we have $d'_{s,t} = \Delta_\lambda(d')(s,t)$ for all $s,t \in S$. Thus, by Theorem~\ref{th:uniquefix} we obtain $d'_{s,t} = \dist(s,t)$ for all $s,t \in S$. \\
(Case 2: $k' \neq k''$ or $m' \neq m''$) By the previous case we have that the optimal value of $D$ is equal to $\sum_{s,t \not\in A} k''_{s,t} + m''_{s,t}$. This is a sum of non negative values, indeed by $\dist(s,t) \geq 0$ for all $s,t \in S$ and Equations~\eqref{eq:m} and \eqref{eq:k}, we have that $k''_{s,t} \geq 0$ and $m''_{s,t} \geq 0$ for all $s,t \notin S$. Thus, if $k' \neq k''$ or $m' \neq m''$ we have $\sum_{s,t \not\in A} k''_{s,t} + m''_{s,t} > \sum_{s,t \not\in A} k'_{s,t} + m'_{s,t}$. This contradicts the assumption that $(d',y',k',m')$ is optimal for $D$.

This proves that if $(d',y',k',m')$ is an optimal solution of $D$ then $d' = \dist$.
\end{proof}
In~\cite{ChenBW12} it has been shown that the bisimilarity distance of Desharnais et al. on discrete-time Markov chains can be computed in polynomial time as the solution of a linear program that can be solved by using the ellipsoid method.
However, in their proposal the number of constraints may be exponential in the size of the model. This is due to the fact that the Kantorovich distance is resolved listing all the couplings that correspond to vertices of the transportation polytopes involved in the definition of the distance~\cite{Demuth61}.

 In contrast, our proposal has a number of constraints and unknowns\footnote{Actually, the variables $k$ are used only for ease the presentation but they can be removed by substitution.} bounded by $3 |S|^2 + |S|^4$ and $|S|^3 + 2 |S|^2$, respectively.
This allows one to use general algorithms for solving LP problems (such as the \emph{simplex} and the \emph{interior-point} methods) that, in practice, are more efficient than the ellipsoid method.

\begin{rem}[The undiscounted case]
Our LP characterization (Theorem~\ref{th:LPchar}) exploits the fact that, for $\lambda \in (0,1)$, the functional operator $\Delta_\lambda$ has unique fixed point (Theorem~\ref{th:uniquefix}). In the so called \emph{undiscounted case}, i.e., when $\lambda = 1$, the fixed point is not unique anymore and the LP characterization does not work directly. To obtain a similar LP characterization with a number of constraints polynomial in the size of the CTMC, by following~\cite{ChenBW12} one may modify the operator $\Delta_1$ to an operator $\Delta'$ that forces the distance to be $0$ if the states are stochastic bisimilar, otherwise it behaves as $\Delta_1$. The operator $\Delta'$ has unique fixed point when $d_L$ is the discrete metric over $L$ (i.e., $d_L(a,b) = 0$ if $a = b$, and $1$ otherwise), but for other choices of $d_L$ the uniqueness may not hold.
\qee
\end{rem}

Moreover, this allows us to state the following complexity result.
\begin{thm} \label{th:distancecomplexity}
$\dist$ can be computed in polynomial-time in the size of $\M$.
\end{thm} 
\begin{proof}
By Theorem~\ref{th:LPchar}, $\dist$ can be computed within the time it takes to construct and solve $D$. $D$ has a number of constraints and unknowns that is bounded by a polynomial in the size of $\M$, therefore its construction can be performed in polynomial time. For the same reason, $D$ admits a polynomial time separation algorithm: whenever a solution is given, its feasibility is checked by scanning each inequality in $D$; otherwise the first encountered inequality that is not satisfied is returned as a separation hyperplane. Therefore, the thesis follows by solving $D$ using the ellipsoid method together with the na\"ive separation algorithm described above.
\end{proof}

\begin{rem}
In Theorem~\ref{th:distancecomplexity} the term ``computed'' should be replaced by ``approximated''. We already discussed about the impossibility of obtaining an \emph{exact} value of $\dist$ (the total variations distance between exponential distributions may assume irrational values!). 

In fact, for the construction of the linear program $D$ one has to use an approximated (rational) version of the coefficients $\rated{s}{t}$, for all $s,t \notin A$, say $e_{s,t}$. To ensure that the optimal solution $d^*$ of $D$ is at most $\varepsilon$ apart from $\dist$ (i.e., $|\dist(s,t) - d^*_{s,t} | \leq \varepsilon$, for all $s,t \in S$) one only needs that $|\rated{s}{t} - e_{s,t} | \leq \frac{\varepsilon}{2}$, for all $s,t \notin A$. 
%
%
This can be done using the Newton-Raphson's iteration algorithm for approximating the $n$-square root of a number $\zeta \in \Q$ up to a precision $\epsilon > 0$. This method is known to be polynomially computable in the size of the representation of $\zeta$, $n$ and $\epsilon$~\cite{Alt79}. Hence $\dist$ can be approximated up to any precision $\varepsilon > 0$ in polynomial-time in the size of $\M$. 
\qee
\end{rem}

\section{Alternative Characterization of the Pseudometric} \label{sec:couplchar}

In the following, we propose an alternative characterization of the bisimilarity distance $\dist$, based on the notion of \emph{coupling structure}. Our result generalizes the one proposed in~\cite{ChenBW12,BacciLM:tacas13} for MCs, to the continuous-time settings.

\begin{defi}[Coupling Structure] \label{def:coupling}
Let $\mathcal{M} = (S, A, \pi, \ell)$ be a CTMC. A \emph{coupling structure} for $\M$ is a function $\C \colon (S\setminus A) \times (S\setminus A) \to \Dist(S \times S)$ such that, for all $s,t \notin A$, $\C(s,t) \in \coupling{\tau(s)}{\tau(t)}$.
\end{defi}
Intuitively, a coupling structure for $\M$ can be thought of as an $S \times S$-indexed collection of joint probability distributions, each having left/right marginals equal to $\tau$.

The following definition adapts the definition of the operator $\Delta_\lambda$ (see Definition~\ref{def:deltaoperator}) with respect to the notion of coupling structure for a CTMC.
\begin{defi}
Let $\M = (S, A, \tau, \rho, \ell)$ be CTMC, $\C$ a coupling structure for $\M$, and $\lambda \in (0,1)$ a \emph{discount factor}. The function $\Gamma^\C_\lambda \colon \mDom \to \mDom$ is defined as follows, for $d \colon S \times S \to [0,1]$ and $s, t \in S$
\begin{equation*}
  \Gamma^\C_\lambda(d)(s,t) =
  \begin{cases}
    1 & \text{if $s \not\equiv_A t$} \\
    \lbld{s}{t} & \text{if $s,t \in A$} \\
    \max \set{\lbld{s}{t}, \lambda \cdot \trcd{s}{t} }{} & \text{if $s,t \notin A$} \\
  \end{cases}
\end{equation*}
where $\lbld{}{}, \rated{}{} \colon S \times S \to [0,1]$ are as in Definition~\ref{def:deltaoperator} and $\trcd[]{}{} \colon \mDom \to \mDom$ is given by
\begin{equation*}
  \textstyle
  \trcd{s}{t} = \rated{s}{t}  + (1 - \rated{s}{t}) \cdot \sum_{u,v \in S} d(u,v) \cdot \C(s,t)(u,v)  \,.
\end{equation*}
\end{defi}
Recall that the Kantorovich distance between two distributions $\mu$ and $\nu$ is defined as $\mathcal{K}_d(\mu,\nu) = \min_\omega \sum_{u,v \in S} d(u,v) \cdot \omega(u,v)$, where the minimum is taken over all the possible couplings $\omega \in \coupling{\mu}{\nu}$. Thus, the operator $\Gamma^\C_\lambda$ can intuitively be thought of as a possible instance of $\Delta^\M_\lambda$ with respect to a fixed choice of the couplings given by $\C$. 

One can easily check that $\Gamma^\C_\lambda$ is monotone, thus, by Tarski's fixed point theorem, it admits least and greatest fixed points. The least fixed point, in particular, will be denoted by $\discr{C}$ and referred to as the \emph{$\lambda$-discrepancy} of $\mathcal{C}$.

%

%
\begin{thm}[Minimum coupling] \label{th:coupling-dist}
$\dist = \min \set{\discr{C}} {\text{$\mathcal{C}$ coupling structure for $\mathcal{M}$}}$.
\end{thm}
\begin{proof}
We first prove that $\delta_\lambda \sqsubseteq \discr{C}$, for any coupling structure $\C$ for $\M$. By Tarski's fixed point theorem, it suffices to prove that, for any $d \colon S \times S \to [0,1]$, $\Delta_\lambda(d) \sqsubseteq \Gamma^\C_\lambda(d)$. The only nontrivial case is when $s,t \notin A$, which follows by definition of $\mathcal{K}_d$, by noticing that $\trd[d]{s}{t} \leq \trcd[d]{s}{t}$ and that the maximum is order preserving.
It remains to prove that the minimum is attained. To this end, define a coupling structure $\C^*$ as $\C^*(s,t) = \omega_{s,t}$, for $s,t \notin A$, where $\omega_{s,t} \in \coupling{\tau(s)}{\tau(t)}$ is such that $\mathcal{K}_{\dist}(\tau(s),\tau(t)) = \sum_{u,v} \dist(u,v) \cdot \omega_{s,t}(u,v)$. By construction, $\dist  = \Gamma^{\C^*}_\lambda(\dist)$, hence $\discr{C^*} \sqsubseteq \dist$. Since $\C^*$ is a coupling structure for $\M$, by what we have shown above we also have $\dist \sqsubseteq \discr{C^*}$. Therefore, $\dist = \discr{C^*}$.
\end{proof}

\section{Greedy Computation of the Bisimilarity Distance} \label{sec:algorithm}

Inspired by the characterization given in Theorem~\ref{th:coupling-dist}, we propose a procedure to compute the bisimilarity pseudometric that is alternative to those previously described. 

The set of coupling structures for $\M$ can be endowed with the preorder $\ordC$ defined as $\C \ordC \C'$ iff $\discr{C} \sqsubseteq \discr{C'}$. 
Theorem~\ref{th:coupling-dist} suggests to look at all the coupling structures $\C$ for $\M$ in order to find an optimal one, i.e., minimal w.r.t.~${\ordC}$. However, it is clear that the enumeration of all the couplings is unfeasible, therefore it is crucial to provide an efficient search strategy which allows one to find an optimal coupling by exploring only a finite amount of them. Moreover we also need an efficient method for computing the $\lambda$-discrepancy associated with a coupling structure.

\subsection{Computing the \texorpdfstring{$\lambda$}{Lambda}-Discrepancy} \label{sec:reachprob}
In this section we consider the problem of computing the $\lambda$-discrepancy associated with a coupling structure. 

By Tarski's fixed point theorem, $\discr{C}$ corresponds to the least pre-fixed point of $\Gamma^\C_\lambda$, that is $\discr{C} = \bigsqcap \{ d \in [0,1]^{S \times S} \mid \Gamma^\C_\lambda(d) \sqsubseteq d \}$. This allows us to compute the $\lambda$-discrepancy associated with $\C$ as the optimal solution of the following linear program, denoted by $\mathit{Discr}_\lambda(\C)$.
\begin{align*}
\operatorname*{arg\,min}_{d} & \: 
\textstyle \sum_{s,t \in S} d_{s,t} \\
	& d_{s,t} \geq 1 && \text{if $s \not\equiv_A t$} \\
	& d_{s,t} \geq \lbld{s}{t} && \text{if $s \equiv_A t$} \\
	& d_{s,t} \geq \lambda \big( \rated{s}{t}  + (1 - \rated{s}{t}) \cdot \textstyle \sum_{u,v \in S} d_{u,v} \cdot \C(s,t)(u,v) \big) && 
	\text{if $s,t \notin A$}
\end{align*}
$\mathit{Discr}_\lambda(\C)$ has a number of inequalities that is bounded by $2 |S|^2$ and $|S|^2$ unknowns, thus, it can be efficiently solved using the interior-point method. 

\begin{rem} \label{rmk:subsystem}
If one is interested in computing the $\lambda$-discrepancy for a particular pair of states $(s,t)$, the method above can be applied on the least independent set of inequalities containing the variable $d_{s,t}$. Moreover, assuming that for some pairs the values associated to $d$ are known, the set of constraints can be further decreased by substitution.
\qee
\end{rem}

\subsection{Greedy Strategy for Optimal Coupling Structures} 
\label{sec:strategy}

In this section, we propose a \emph{greedy strategy} that moves toward an optimal coupling structure starting from any given one. Then, we provide sufficient and necessary conditions for a coupling structure, to ensure that its associated $\lambda$-discrepancy coincides with $\dist$.

Hereafter we fix a CTMC $\M =(S, A, \tau, \rho, \ell)$ and a coupling structure $\C$ for it. 
The greedy strategy takes a coupling structure and locally updates it at a given pair of states in such a way that it decreases it with respect to $\ordC$. 
For $s,t \notin A$ and $\omega \in \coupling{\tau(s)}{\tau(t)}$, we denote by $\C[(s,t)/\omega]$ the \emph{update} of $\C$ at $(s,t)$ with $\omega$, defined as $\C[(s,t)/\omega](u,v) = \C(u,v)$, for all $(u,v) \neq (s,t)$, and $\C[(s,t)/\omega](s,t) = \omega$; it is worth noting that, by construction, $\C[(s,t)/\omega]$ is a coupling structure of $\M$.

The next lemma gives a sufficient condition for an update to be effective for the strategy.
\begin{lem} \label{lem:betterapprox}
Let $s, t \notin A$ and $\omega \in \coupling{\tau(s)}{\tau(t)}$. Then, for $\mathcal{H} = \C[(s,t)/\omega]$ and any $\lambda \in (0,1)$, if $\Gamma^{\mathcal{H}}_\lambda(\discr{C})(s,t) < \discr{C}(s,t)$ then $\discr{H} \sqsubset \discr{C}$.
\end{lem}
\begin{proof} 
It suffices to show that $\Gamma^{\mathcal{H}}_\lambda (\discr{C}) \sqsubset \discr{C}$, i.e., that $\discr{C}$ is a strict post-fixed point of $\Gamma^\mathcal{H}_\lambda$. Then, the thesis follows by Tarski's fixed point theorem. 

Let $u,v \in S$. If $u \not\equiv_A v$, then $\Gamma^{\mathcal{H}}_{\lambda}(\discr{C})(u,v) = 1 = 
\Gamma^\C_\lambda(\discr{C})(u,v) =\discr{C}(u,v)$. If $u,v \in A$, then $\Gamma^{\mathcal{H}}_{\lambda}(\discr{C})(u,v) = \lbld{u}{v} = \Gamma^\C_\lambda(\discr{C})(u,v) =\discr{C}(u,v)$. If $u,v \notin A$ and $(u,v) \neq (s,t)$, by definition of $\mathcal{H}$, we have that $\C(u,v) = \mathcal{H}(u,v)$, hence
$\Gamma^\C_{\lambda}(\discr{C})(u,v) = \Gamma^{\mathcal{H}}_{\lambda}(\discr{C})(u,v)$. The remaining case, i.e., $(u,v) = (s,t)$, holds by hypothesis. This proves $\Gamma^{\mathcal{H}}_\lambda (\discr{C}) \sqsubset \discr{C}$.
\end{proof}

Lemma~\ref{lem:betterapprox} states that $\mathcal{C}$ can be improved w.r.t. $\ordC$ by updating it at $(s,t)$, if $s,t \notin A$ and there exists a coupling $\omega \in \coupling{\tau(s)}{\tau(t)}$ such that the following holds
\begin{equation*}
  \textstyle
  \sum_{u,v \in S} \discr{C}(u,v) \cdot \omega(u,v) <
  \sum_{u,v \in S} \discr{C}(u,v) \cdot \C(s,t)(u,v) \,.
\end{equation*}
A coupling that enjoys the above condition is $\omega \in TP(\discr{C}, \tau(s), \tau(t))$ where, for arbitrary $\mu,\nu \in \Dist(S)$ and $c \colon S \times S \to [0,1]$
\begin{equation}
\begin{aligned}
TP(c, \mu, \nu) = \operatorname*{arg\,min}_{\omega} & \: 
\textstyle \sum_{s,t \in S} c(u,v) \cdot \omega_{u,v} \\[-0.5ex]
	&\textstyle \sum_{v} \omega_{u,v} = \mu(u) && \forall u \in S \\
	&\textstyle \sum_{u} \omega_{u,v} = \nu(v) && \forall v \in S \\
	& \omega_{u,v} \geq 0 && \forall u,v \in S \, .
\end{aligned}
\label{eq:transportationproblem}
\end{equation}
The above problem is usually referred to as the (homogeneous) \emph{transportation problem} with $\mu$ and $\nu$ as left and right marginals, respectively, and transportation costs $c$. This problem has been extensively studied and comes with (several) efficient polynomial algorithmic solutions~\cite{Dantzig51,FordF56}.

This gives us an efficient solution to update any coupling structure, that, together with Lemma~\ref{lem:betterapprox} represents a strategy for moving toward $\dist$ by successive improvements on the coupling structures.

Now we proceed giving a sufficient and necessary condition for termination.
\begin{lem} \label{lem:bettercoupling}
If $\discr{C} \neq \dist$, then there exist $s,t \notin A$ and a coupling structure $\mathcal{H} = \mathcal{C}[(s,t)/\omega]$ for $\mathcal{M}$ such that $\Gamma^{\mathcal{H}}_{\lambda}(\discr{C})(s,t) < \discr{C}(s,t)$.
\end{lem}
\begin{proof}
We proceed by contraposition. If for all $s,t \notin A$ and $\omega \in \coupling{\tau(s)}{\tau(t)}$, $\Gamma^\mathcal{H}_\lambda(\discr{C})(s,t) \geq \discr{C}(s,t)$, then $\discr{C} = \Delta_\lambda(\discr{C})$. Since, by Theorem~\ref{th:uniquefix}, $\Delta_\lambda$ has a unique fixed point, $\discr{C} = \dist$.
\end{proof}
The above result ensures that, unless $\mathcal{C}$ is optimal w.r.t\, 
$\ordC$, the hypothesis of Lemma~\ref{lem:betterapprox} is 
satisfied, so that, we can further improve $\mathcal{C}$ as aforesaid.

The next statement proves that this search strategy is correct.
\begin{thm} \label{th:correctness}
$\dist = \discr{C}$ iff there is no coupling $\mathcal{H}$ for
$\mathcal{M}$ such that $\Gamma^{\mathcal{H}}_{\lambda}(\discr{C}) \sqsubset \discr{C}$.
\end{thm}
\begin{proof}
We prove: $\dist \neq \discr{C}$ iff 
there exists $\mathcal{H}$ such that $\Gamma^{\mathcal{H}}_{\lambda}(\discr{C}) \sqsubset \discr{C}$.
($\Rightarrow$) Assume $\dist \neq \discr{C}$. 
By Lemma~\ref{lem:bettercoupling}, there exist a pair of states $s,t \in S$ and a coupling $\omega \in \coupling{\tau(s)}{\tau(t)}$ such that $\lambda \cdot \sum_{u,v \in S} \discr{C}(u,v) \cdot \omega(u,v) < \discr{C}(s,t)$. As in the proof of Lemma~\ref{lem:betterapprox}, we have that $\mathcal{H} = \mathcal{C}[(s,t)/\omega]$ 
satisfies $\Gamma^{\mathcal{H}}_\lambda (\discr{C}) \sqsubset \discr{C}$.
($\Leftarrow$) Let $\mathcal{H}$ be such that 
$\Gamma^{\mathcal{H}}_{\lambda}(\discr{C}) \sqsubset \discr{C}$. 
By Tarski's fixed point theorem $\discr{D} \sqsubset \discr{C}$. 
By Theorem~\ref{th:coupling-dist}, $\dist \sqsubseteq \discr{D} \sqsubset \discr{C}$.
\end{proof}

\begin{rem} \label{rmk:vertexes}
Note that, in general there could be an infinite number of coupling structures for a given CTMC. However, for each fixed $d \in \mDom$, the linear function mapping $\omega$ to $\sum_{u,v \in S} d(u,v) \cdot \omega(u,v)$ achieves its minimum at some vertex in the transportation polytope $\coupling{\tau(s)}{\tau(t)}$. 
Since the number of such vertices is finite, the termination of the search strategy is ensured by updating the coupling structure using optimal vertices. This does not introduce further complications in the algorithm since there are methods for solving the transportation problem (e.g., Dantzig's primal simplex method~\cite{Dantzig51}) which provide optimal transportation schedules that are vertices. \qee
\end{rem}

\section{The On-the-Fly Algorithm} \label{sec:onthefly}
\renewcommand{\discr}[2][\lambda]{\gamma_{#1}^{#2}}
\newcommand{\demanded}[2][\mathit{Exact},\C]{\mathcal{R}_{#2}(#1)}
In this section we describe an on-the-fly technique for computing the bisimilarity distance $\dist$ fully exploiting the greedy strategy of Section~\ref{sec:strategy}.

Let $Q \subseteq S \times S$ and consider the problem of computing $\dist(s,t)$ for all $(s,t) \in Q$. Recall that the strategy proposed in Section~\ref{sec:strategy} consists in a traversal $\C_0 \triangleright_\lambda \C_1 \triangleright_\lambda \cdots \triangleright_\lambda  \C_n$ of the set of coupling structures for $\M$ that starts from an arbitrary coupling structure $\C_0$ and leads to an optimal one $\C_n$. 
We observe that, for any $i < n$
\begin{enumerate}
  \item the improvement of each coupling structure $\C_i$ is obtained by a \emph{local update} at some pair of states $u,v \notin A$, namely $\C_{i+1} = \C_{i}[(u,v)/\omega]$ for some $\omega \in TP(\discr{\C_i},\tau(u),\tau(v))$;
  \item the pair $(u,v)$ is chosen according to an optimality check that is performed \emph{locally} among the couplings in $\coupling{\tau(u)}{\tau(v)}$, i.e., $\C_i(u,v) \notin TP(\discr{\C_i},\tau(u),\tau(v))$;
  \item whenever a coupling structure $\C_i$ is considered, its associated $\lambda$-discrepancy $\discr{\C_i}$ can be computed by solving the linear program $\mathit{Discr}_\lambda(\C_i)$ described in Section~\ref{sec:reachprob}.
\end{enumerate}
Among the observations above, only the last one requires to look at the coupling structure $\C_i$.
However, as noticed in Remark~\ref{rmk:subsystem}, the value $\discr{\C_i}(s,t)$ can be computed without considering the entire set of constraints of $\mathit{Discr}_\lambda(\C_i)$, but only the least independent set of inequalities that contains the variable $d_{s,t}$. 
Moreover, provided that for some pairs of states $E \subseteq S \times S$ the value of the distance is known, the linear program $\mathit{Discr}_\lambda(\C_i)$ can be further reduced by substituting the occurrences of the unknown $d_{u,v}$ by the constant $\dist(u,v)$, for each $(u,v) \in E$.
This suggests that we do not need to store the entire coupling structures, but they can be constructed \emph{on-the-fly} during the calculation.
Specifically, the couplings that are demanded to compute $\discr{\C_i}(s,t)$ are only those $\C_i(u,v)$ such that \mbox{$(s,t) \leadsto_{\C_i,E}^* (u,v)$}, where $\leadsto_{\C_i,E}^*$ is the reflexive and transitive closure of ${\leadsto_{\C_i,E}} \subseteq S^2 \times S^2$, defined by
\begin{equation*}
(s',t') \leadsto_{\C_i,E} (u',v') \quad \text{ iff } \quad \C_i(s',t')(u',v') > 0 \text{ and } (u',v') \notin E \,.
\end{equation*}

\begin{algorithm}[t]
    \algsetup{linenodelimiter=.}
    \caption{On-the-Fly Bisimilarity Pseudometric}
    \begin{algorithmic}[1]  
    \REQUIRE CTMC $\M = (S, A, \tau, \rho, \ell)$;
    discount factor $\lambda \in (0,1)$; query $Q \subseteq S \times S$.
    \STATE $\mathcal{C} \gets$ empty; $Dom_\C \gets \emptyset$; $d \gets$ empty; \label{line:init1}
    	\COMMENT{initialize data structures}
    \STATE $Visited \gets \emptyset$; $Exact \gets \emptyset$; 
    	$ToCompute \gets Q$ \label{line:init2}
    \WHILE{$ToCompute\neq \emptyset$} \label{mainloop1}
    \STATE pick $(s,t) \in ToCompute$ \label{pickone}
   
	
    \IF{$s \not\equiv_A t$} \label{trivial1}
    \STATE $d(s,t) \gets 1$; $Exact \gets Exact \cup \{(s,t)\}$; $Visited \gets Visited \cup \{(s,t)\}$
    \ELSIF{$s = t$}
    \STATE $d(s,t) \gets 0$; $Exact \gets Exact \cup \{(s,t)\}$; $Visited \gets Visited \cup \{(s,t)\}$
    	 \label{trivial2}
    \ELSIF{$s,t \in A$}
    \STATE $d(s,t) \gets \lbld{s}{t}$; $Exact \gets Exact \cup \{(s,t)\}$; $Visited \gets Visited \cup \{(s,t)\}$
    	 \label{trivial3}	 
    \ELSE[if $(s,t)$ is nontrivial] 
    \IF[if $(s,t)$ has not been encountered so far]{$(s,t) \notin Visited$}
    	\STATE pick $\omega \in \coupling{\tau(s)}{\tau(t)}$
		\COMMENT{guess a coupling}
	\STATE $\mathit{SetPair}(\mathcal{M}, (s,t), \omega)$ 
		\COMMENT{update the current coupling structure} \label{buildC}
    \ENDIF    
    \STATE $\mathit{Discrepancy}(\lambda,(s,t))$
    	\COMMENT{update $d$ as the $\lambda$-discrepancy for $\mathcal{C}$} \label{discrC}
     \WHILE{$\exists (u,v) \in \demanded{s,t} \text{ such that  } \mathcal{C}(u,v) \notin \mathit{TP}(d,\tau(u), \tau(v))$} \label{improve1}
    	\STATE $\omega \in \mathit{TP}(d,\tau(u), \tau(v))$
		\COMMENT{pick an optimal coupling for $s,t$ w.r.t.\ $d$}
    	\STATE $\mathit{SetPair}(\mathcal{M}, (u,v), \omega)$
		\COMMENT{improve the current coupling structure}
	\STATE $\mathit{Discrepancy}(\lambda,(s,t))$ 
		\COMMENT{update $d$ as the $\lambda$-discrepancy for $\mathcal{C}$}
    \ENDWHILE \label{improve2}
    \STATE $Exact \gets Exact \cup \demanded{s,t}$ \label{addnewexact}
    	\COMMENT{add new exact distances}
    \STATE remove from $\C$ all the couplings associated with a pair in $Exact$   \label{flushC}
    \ENDIF 
    \STATE $ToCompute \gets ToCompute  \setminus Exact$ 
    	\COMMENT{remove exactly computed pairs} \label{updatequery}
    \ENDWHILE \label{mainloop2}
    \RETURN $d{\restriction}_{Q}$
    \COMMENT{return the distance restricted to all pairs in $Q$} \label{return}
    \end{algorithmic}
    \label{alg:genericBisimPseudoSimm}
\end{algorithm}

The computation of the bisimilarity pseudometric is implemented by 
Algorithm~\ref{alg:genericBisimPseudoSimm}. It takes as input a finite CTMC 
$\M = (S, A, \tau, \rho, \ell)$, a discount factor $\lambda \in (0,1)$, and a query set $Q \subseteq S \times S$.
We assume the following global variables to store: 
\begin{itemize}
  \item $\mathcal{C}$: the current (partial) coupling structure;
  \item $d$: the $\lambda$-discrepancy associated with $\mathcal{C}$;
  \item $ToCompute$: the pairs of states for which the distance has to be computed; 
  \item $Exact$: the set pairs of states $(s,t)$ such that $d(s,t) = \dist(s,t)$, i.e., those pairs which do not need to be further improved\footnote{Actually, the set $Exact$ contains those pairs such that $|d(s,t) - \dist(s,t)| \leq \epsilon$ where $\epsilon$ corresponds to the precision of the machine. In our implementation $\epsilon = 10^{-9}$.};
  \item $Visited$: the set of pairs of states that have been visited so far.
\end{itemize}
Moreover, $\demanded{s,t}$ will denote the set $\{ (u,v) \mid (s,t) \leadsto_{\C,Exact}^* (u,v) \}$.

At the beginning  (line~\ref{line:init1}--\ref{line:init2}) both the coupling structure $\mathcal{C}$ and the 
discrepancy $d$ are empty, there are no visited states, no exact computed distances, and the pairs to be computed are those in the input query.

While there are still pairs left to be computed (line~\ref{mainloop1}), we pick one (line~\ref{pickone}), say $(s,t)$. According to the definition of $\dist$, if $s \not\equiv_A t$ then $\dist(s,t)=1$; if $s = t$ then $\dist(s,t) = 0$ and if $s,t \in A$ then $\dist(s,t) = \lbld{s}{t}$, so that, $d(s,t)$ is set accordingly, and $(s,t)$ is added to $\mathit{Exact}$ (lines~\ref{trivial1}--\ref{trivial3}). Otherwise, if $(s,t)$ was not previously visited, a coupling $\omega \in \coupling{\tau(s)}{\tau(t)}$ is guessed, and the routine $\mathit{SetPair}$ updates the coupling structure $\C$ at $(s,t)$ with $\omega$ (line~\ref{buildC}), then the routine $\mathit{Discrepancy}$ updates $d$ with the $\lambda$-discrepancy associated with $\C$ (line~\ref{discrC}). According to the greedy strategy, $\C$ is successively improved and $d$ is consequently updated, until no further improvements are possible (lines~\ref{improve1}--\ref{improve2}). Each improvement is obtained by replacing a sub-optimal coupling $C(u,v)$, for some $(u,v) \in \demanded{s,t}$, by one taken from $TP(d, \tau(u), \tau(v))$ (line~\ref{improve1}).
Note that, each improvement actually affects the current value of $d(s,t)$, since the update is performed on a pair in $\demanded{s,t}$. It is worth to note that $\C$ and $\mathit{Exact}$ are constantly updated, hence $\demanded{s,t}$ may differ from one iteration to another.

When line~\ref{addnewexact} is reached, for each $(u,v) \in \demanded{s,t}$, we are guaranteed that $d(u,v) = \dist(s,t)$, therefore $\demanded{s,t}$ is added to $\mathit{Exact}$ and, for these pairs, $d$ will no longer be updated. At this point (line~\ref{flushC}), the couplings associated with the pairs in $\mathit{Exact}$ can be removed from $\C$.
In line~\ref{updatequery}, the exact pairs computed so far are removed from $ToCompute$. 
Finally, if no more pairs need be considered, the exact distance on $Q$ is returned (line~\ref{return}).

Algorithm~\ref{alg:genericBisimPseudoSimm} calls the subroutines $\mathit{SetPair}$ and
$\mathit{Discrepancy}$. The former is used to construct and update the coupling structure $\mathcal{C}$,
the latter to update the current over-approxi\-ma\-tion $d$ during the computation.
Next, we explain how they work. 

\begin{algorithm}[t]
    \algsetup{linenodelimiter=.}
    \caption{$\mathit{SetPair}(\mathcal{M}, (s,t), \omega)$} 
    \begin{algorithmic}[1]  
    \REQUIRE  CTMC $\M = (S, A, \tau, \rho, \ell)$; $s,t \in S$; $\omega \in \coupling{\tau(s)}{\tau(t)}$
    \STATE $\mathcal{C}(s,t) \gets \omega$
    	\COMMENT{update the coupling at $(s,t)$ with $\omega$}  \label{setmatching1}
    \STATE $Visited \gets Visited \cup \{(s,t)\}$ \COMMENT{set $(s,t)$ as visited}  \label{setmatching2}
    \FORALL[for all demanded pairs]{$(u,v) \notin \mathit{Visited}$ such that 
    $(s,t) \leadsto_{\C,Exact} (u,v)$} \label{recbuild1}
    \STATE $Visited \gets Visited \cup \{(u,v)\}$
    \STATE \textbf{if} $u = v$ \textbf{then} 
    	$d(u,v) \gets 0$; $Exact \gets Exact \cup \{(u,v)\}$;  \label{settrivial1}
    \STATE \textbf{if} $u \not\equiv_A v$ \textbf{then} 
    	$d(u,v) \gets 1$; $Exact \gets Exact \cup \{(u,v)\}$;  \label{settrivial2}
    \STATE \textbf{if} $u,v \in A$ \textbf{then} 
    	$d(u,v) \gets \lbld{u}{v}$; $Exact \gets Exact \cup \{(u,v)\}$;  \label{settrivial3}
    \STATE // propagate the construction
    \IF{$(u,v) \notin Exact$}
    \STATE pick $\omega' \in \coupling{\tau(u)}{\tau(v)}$ 
    	\COMMENT{guess a matching}
    \STATE $\mathit{SetPair}(\mathcal{M}, (u,v), \omega')$
    \ENDIF
    \ENDFOR \label{recbuild2}
    \end{algorithmic}
    \label{alg:SetPair}
\end{algorithm}

$\mathit{SetPair}$ (Algorithm~\ref{alg:SetPair}) takes as input a CTMC $\M = (S, A, \tau, \rho, \ell)$, a pair of states $s,t \in S$, and a coupling $\omega \in \coupling{\tau(s)}{\tau(t)}$. In lines~\ref{setmatching1}--\ref{setmatching2}, the coupling structure $\C$ is set to $\omega$ at $(s,t)$, then $(s,t)$ is added to $\mathit{Visited}$. The on-the-fly construction of the coupling structure is recursively propagated to the demanded successor pairs of $(s,t)$ according to the information accumulated so far. 
During this construction, if some states with trivial distances are encountered, $d$ and
$Exact$ are updated accordingly (lines~\ref{settrivial1}--\ref{settrivial3}).

\begin{algorithm}[t]
    \algsetup{linenodelimiter=.}
    \caption{$\mathit{Discrepancy}(\lambda,(s,t))$}
    \begin{algorithmic}[1]
    \REQUIRE discount factor $\lambda \in (0,1)$; $s,t \in S \setminus A$
    
    \STATE Let LP be the linear program obtained from $\mathit{Discr}_\lambda(\C)$ by keeping only the    inequalities associated with pairs in $\demanded{s,t}$ and replacing the unknown $d_{u,v}$ by the constant $d(u,v)$, for all $(u,v) \in \textit{Exact}$. \label{smallestLP}
    \STATE $d^* \gets$ optimal solution of LP
    \FORALL[update distances]{$(u,v) \in \demanded{s,t}$} \label{discrupdateloop1}
    \STATE $d(u,v) \gets d^{*}_{u,v}$ 
    \IF{$d(u,v) = 0$ or $d(u,v) = \lbld{u}{v}$}
    \STATE $\textit{Exact} \gets \textit{Exact} \cup \{(u,v)\}$
    \ENDIF \label{discrupdateloop2}
    \ENDFOR 
    \end{algorithmic}
    \label{alg:discrepancy}
\end{algorithm}

$\mathit{Discrepancy}$ (Algorithm~\ref{alg:discrepancy}) takes as input
a discount factor $\lambda \in (0,1)$ and a pair of states $s,t \notin A$. It constructs the least linear program obtained from $\mathit{Discr}_\lambda(\C)$, that can compute $\discr{\C}(s,t)$ using the information accumulated so far (line~\ref{smallestLP}). In lines~\ref{discrupdateloop1}--\ref{discrupdateloop2} the current $\lambda$-discrepancy is updated accordingly; and those pairs $(u,v) \in \demanded{s,t}$ for which the current $\lambda$-discrepancy coincide with the distance are added to $\mathit{Exact}$.

\begin{figure}[t]
\centering
\begin{tikzpicture}[
  label1/.style={circle, draw=red!50,fill=red!20, thin, circle split, inner sep=0.6mm},
  label2/.style={circle,draw=blue!50,fill=blue!20, thin, circle split, inner sep=0.6mm},
  label3/.style={circle,draw=green!50!black,fill=green!80!blue, thin, circle split, inner sep=0.6mm}]
 
\draw 
  (0,0) node[label1] (s1) {$s_1$ \nodepart{lower} $15$}
  ($(s1) +(right:2.5)$) node[label2] (s2) {$s_2$ \nodepart{lower} $9$}
  ($(s1) +(right:2.5)+(down:3)$) node[label1] (s3) {$s_3$ \nodepart{lower} $15$}
  ($(s1) +(down:3)$) node[label3] (s4) {$s_4$  \nodepart{lower} $9$}
  ;

\path[-latex]
  (s1) edge[bend left=20]  node[above] {$\frac{5}{7}$} (s2)
  (s1) edge node[left] {$\frac{2}{7}$} (s4)%
  (s2) edge[bend left=20] node[below] {$\frac{1}{4}$} (s1)
  (s2) edge[loop above] node[pos=0.8, right] {$\frac{3}{4}$} (s2) %
  (s3) edge[bend left=20] node[below] {$\frac{1}{2}$} (s4)
  (s3) edge node[right] {$\frac{1}{2}$} (s2) %
  (s4) edge node[pos=0.3, above] {$\frac{1}{9}$} (s2)
  (s4) edge[bend left=20] node[pos=0.6, above] {$\frac{4}{9}$} (s3)
  (s4) edge[loop below] node[pos=0.8, left] {$\frac{4}{9}$} (s4)%
  ;
%
\tikzset{ 
table/.style={
  matrix of math nodes,
  row sep=-\pgflinewidth,
  column sep=-\pgflinewidth,
  nodes={rectangle, draw, text width=2.7ex, text height=2ex, align=center},
  text depth=0.8ex,
  text height=2ex,
  nodes in empty cells,
  row 1/.style={draw=white}, column 1/.style={draw=white} }
}

\def\namesep{5pt}

\draw ($(s2)+(right:1.5cm)+(down:0.2cm)$) node[matrix, table, anchor=west] (c014) {
   & s_2 & s_3 & s_4 \\
  s_2 & & \frac{4}{9} & \frac{17}{63} \\
  s_4 & \frac{1}{9} &  & \frac{11}{63} \\
};
\draw ($(c014-2-2.north west) $) 
node[anchor=south east, inner sep=1pt] {$\omega_{1,4}\colon$};

\draw ($(c014.east)+(right:0.4cm)$) node[matrix, table, anchor=west] (c012) {
   & s_1 & s_2 \\
  s_2 & & \frac{5}{7} \\
  s_4 & \frac{1}{4} & \frac{1}{28} \\
};
\draw ($(c012-2-2.north west)$) 
node[anchor=south east, inner sep=1pt] {$\omega_{1,2}\colon$};

\draw ($(c012.south)$) node[matrix, table, anchor=north] (c023) {
   & s_2 & s_4 \\
  s_1 & \frac{1}{4} & \\
  s_2 & \frac{1}{4} & \frac{1}{2} \\
};
\draw ($(c023-2-2.north west)$) 
node[anchor=south east, inner sep=1pt] {$\omega_{2,3}\colon$};

\draw ($(c014.south)$) node[matrix, table, anchor=north] (c024) {
   & s_2 & s_3 & s_4 \\
  s_1 & \frac{1}{9} & & \frac{5}{36} \\
  s_2 &  & \frac{4}{9} & \frac{11}{36} \\
};
\draw ($(c024-2-2.north west)$) 
node[anchor=south east, inner sep=1pt] {$\omega_{2,4}\colon$};

\path[thick] 
($(c014.north west)+(left:0.5)$) edge node[fill=white] {$\C_0$} ($(c012.north east)$)
($(c024.south west)+(left:0.5)$) edge ($(c023.south east)$)
($(c014.north west)+(left:0.5)$) edge ($(c024.south west)+(left:0.5)$)
($(c012.north east)$) edge ($(c023.south east)$)
;


\draw ($(c012.east)+(right:1cm)$) node[matrix, table, anchor=west] (c114) {
   & s_2 & s_3 & s_4 \\
  s_2 & \frac{1}{9} & \frac{4}{9} & \frac{10}{63} \\
  s_4 & &  & \frac{2}{7} \\
};
\draw ($(c114-2-2.north west)$) 
node[anchor=south east, inner sep=1pt] {$\omega'_{1,4}\colon$};

\path[thick] 
($(c114.north west)+(left:0.5)$) edge node[fill=white] {$\C_1$} ($(c114.north east)$)
($(c114.south west)+(left:0.5)$) edge ($(c114.south east)$)
($(c114.north west)+(left:0.5)$) edge ($(c114.south west)+(left:0.5)$)
($(c114.north east)$) edge ($(c114.south east)$)
;

\newcommand{\lblbox}[1]{\tikz {\node[fill=#1] {};} }

\draw ($(c114.south) +(down:1ex)+(left:1ex)$) node[anchor=north, matrix, matrix of math nodes] {
  d_L(\lblbox{red!50}, \lblbox{blue!50}) = \frac{1}{2} \\
  d_L(\lblbox{red!50}, \lblbox{green!80!blue}) = \frac{1}{6} \\
  d_L(\lblbox{blue!50}, \lblbox{green!80!blue}) = \frac{2}{3} \\
};

\end{tikzpicture}

\caption{Execution trace for the computation of $\dist[\frac{1}{2}](1,4)$ (details in Example~\ref{ex:ontheflyexample}).}
\label{fig:example}
\end{figure}

Next, we present a simple example of Algorithm~\ref{alg:genericBisimPseudoSimm}, 
showing the main features of our method:
(1) the on-the-fly construction of the (partial) coupling, and (2) the restriction only to those variables 
which are demanded for the solution of the system of linear equations.
\begin{exa}[On-the-fly computation] \label{ex:ontheflyexample}
Consider the CTMC in Figure~\ref{fig:example}, and assume we want to compute the $\lambda$-discounted bisimilarity distance between states $s_1$ and $s_4$, for $\lambda = \frac{1}{2}$.

Algorithm~\ref{alg:genericBisimPseudoSimm} starts by guessing an initial coupling structure $\mathcal{C}_0$. This is done by considering only the pairs of states which are really needed in the computation. Starting from the pair $(s_1,s_4)$ a coupling in $\omega_{1,4} \in \coupling{\tau(s_1)}{\tau(s_4)}$ is guessed as in Figure~\ref{fig:example} and assigned to $\C_0(s_1,s_4)$. 
This demands for the exploration of the pairs  $(s_2,s_3)$, $(s_2,s_4)$, $(s_1,s_2)$ and the guess of three new couplings $\omega_{2,3} \in \coupling{\tau(s_2)}{\tau(s_3)}$, $\omega_{2,4} \in \coupling{\tau(s_2)}{\tau(s_4)}$, and $\omega_{1,2} \in \coupling{\tau(s_1)}{\tau(s_2)}$, to be associated in $\C_0$ with their corresponding pairs. Since no other pairs are demanded, the construction of $\C_0$ terminates as shown in Figure~\ref{fig:example}. The $\lambda$-discrepancy associated with $\C_0$ for the pair $(s_1, s_4)$ is obtained as the solution of the following reduced linear program  
\begin{align*}
\operatorname*{arg\,min}_{d} & \: ( d_{1,4} + d_{2,3} + d_{2,4} + d_{1,2} ) \\
&d_{1,4} \geq \frac{1}{6} \\[-2ex]
&d_{1,4} \geq \frac{\alpha}{2} + \frac{(1-\alpha)}{2} \cdot 
	\Big( \frac{4}{9} \cdot d_{2,3} + \frac{17}{63} \cdot d_{2,4} + \frac{11}{63} \cdot \overbrace{d_{4,4}}^{{}=0} \Big) \\
&d_{2,3} \geq \frac{1}{2} \\[-2ex]
&d_{2,3} \geq \frac{\alpha}{2} + \frac{(1-\alpha)}{2} \cdot 
	 \Big( \frac{1}{4} \cdot d_{1,2} + \frac{1}{4} \cdot \overbrace{d_{2,2}}^{{}=0} + \frac{1}{2} \cdot d_{2,4} \Big) \\
&d_{2,4} \geq \frac{2}{3} \\
&d_{2,4} \geq \frac{1}{2} \cdot 
	\Big( \frac{1}{9} \cdot d_{1,2} + \frac{5}{36} \cdot d_{1,4} + \frac{4}{9} \cdot d_{2,3} + \frac{11}{36} \cdot d_{2,4} \Big) \\
&d_{1,2} \geq \frac{1}{2} \\[-2ex]
&d_{1,2} \geq \frac{\alpha}{2} + \frac{(1-\alpha)}{2} \cdot 
	\Big( \frac{5}{7} \cdot \overbrace{d_{2,2}}^{{}=0} + \frac{1}{4} \cdot d_{1,4} + \frac{1}{28} \cdot d_{2,4} \Big)
\end{align*}
where $\alpha = \tv{\Exp[15]}{\Exp[9]} = \frac{\sqrt[6]{3/5}}{25}$ (by Equation~\eqref{eq:tvexp}). Note that, the bisimilarity distance for the pairs $(s_2,s_2)$ and $(s_4,s_4)$ is always $0$, thus $d_{2,2}$ and $d_{4,4}$ are substituted accordingly. The solution of the above linear program is 
$d^{\C_0}(s_1,s_4) = \frac{\alpha}{2} + \frac{5(1-\alpha)}{21}$, $d^{\C_0}(s_2,s_3) = \frac{1}{2}$, $d^{\C_0}(s_2,s_4) = \frac{2}{3}$, and $d^{\C_0}(s_1,s_2) = \frac{1}{2}$.

Since, the $\lambda$-discrepancy for $(s_2,s_3)$, $(s_2,s_4)$, and $(s_1,s_2)$ equals the distance $\lbld{}{}$ between their labels, it coincides with the bisimilarity distance, hence it cannot be further decreased. Consequently, the pairs of states are added to the set $\mathit{Exact}$ and their associated couplings are removed from $\C_0$. Note that, these pairs will no longer be considered in the construction of a coupling structure.

In order to decrease the $\lambda$-discrepancy of $(s_1,s_4)$, Algorithm~\ref{alg:genericBisimPseudoSimm} constructs a new coupling structure $\C_1$. According to our greedy strategy, $\C_1$ is obtained from $\C_0$ updating $\C_0(s_1,s_4)$ (i.e., the only coupling left) by the coupling $\omega'_{1,4} \in \coupling{\tau(s_1)}{\tau(s_4)}$ (shown in Figure~\ref{fig:example}) that is obtained as the solution of a transportation problem with marginals $\tau(s_1)$ and $\tau(s_4)$, where the current $\lambda$-discrepancy is taken as cost function.
The resulting coupling does not demand for the exploration of new pairs in the CTMC, hence the construction of $\C_1$ terminates. The reduced linear program associated with $\C_1$ is given by 
\begin{align*}
\operatorname*{arg\,min}_{d} & \: d_{1,4} \\
&d_{1,4} \geq \frac{1}{6} \\[-2ex]
&d_{1,4} \geq \frac{\alpha}{2} + \frac{(1-\alpha)}{2} \cdot \Big( 
	\frac{1}{9} \cdot \overbrace{d_{2,2}}^{{}=0}  + 
	\frac{4}{9} \cdot \overbrace{d_{2,3}}^{{}=\frac{1}{2}}  + 
	\frac{10}{63} \cdot \overbrace{d_{2,4}}^{{}=\frac{2}{3}}  + 
	\frac{2}{7} \cdot \overbrace{d_{4,4}}^{{}=0}
  \Big)
\end{align*}
whose solution is $d^{\C_1}(s_1,s_4) = \frac{\alpha}{2} + \frac{31(1-\alpha)}{189}$.

Solving again a new transportation problem with the improved current $\lambda$-discrepancy as cost function, we discover that the coupling structure $\C_1$ cannot be further improved, hence we stop the
computation, returning $\dist(s_1,s_4) = d^{\C_1}(s_1,s_4) = \frac{\alpha}{2} + \frac{31(1-\alpha)}{189}$.
\qee
\end{exa}

\begin{rem} \label{rmk:overapproximatedusage}
Algorithm~\ref{alg:genericBisimPseudoSimm} can also be used for computing over-appro\-xi\-mated 
distances. Indeed, assuming over-estimates for some particular distances are already known, they
can be taken as inputs and used in our algorithm simply storing them in the variable $d$ and treated
as ``exact'' values. In this way our method will return the least over-approximation
of the distance agreeing with the given over-estimates.
This modification of the algorithm can be used to further decrease the exploration of the CTMC.
Moreover, it can be employed in combination with approximated algorithms, having
the advantage of an on-the-fly state space exploration.\qee
\end{rem}

\section{Experimental Results} \label{sec:experiments}

In this section, we evaluate the performance of the on-the-fly algorithm on a 
collection of randomly generated CTMCs\footnote{
The tests have been performed on a prototype implementation coded
in $\text{Wolfram Mathematica}^{\circledR}\, 9$ (available at 
\url{http://people.cs.aau.dk/~giovbacci/tools.html})
running on an Intel Core-i7 3.4 GHz processor with 12GB of RAM.}.

\begin{table}[t]
\begin{center}
\footnotesize
\setlength{\tabcolsep}{1ex}
\begin{tabular}{|c|c|c|c|c|c|}
\hline
 \multirow{2}{*}{\# States} & \multicolumn{2}{c|}{On-the-Fly (exact)} & 
 \multicolumn{2}{c|}{Iterating (approximated)} & Approx.\ \\[0.3ex] \cline{2-5}
  & Time (s) & \# TPs & \# Iterations & \# TPs &  Error \\[0.3ex] \hline\hline 
 10 & 0.352 & 10.500 & 2.660 & 266.667 & 0.0339\\
 12 & 0.772 & 19.700 & 2.850 & 410.403 & 0.0388 \\
 14 & 2.496 & 35.800 & 3.880 & 760.480 & 0.0318 \\
 16 & 4.549 & 50.607 & 5.142 & 1316.570 & 0.0230 \\
 18 & 13.709 & 78.611 & 6.638 & 2151.021 & 0.0206 \\
 20 & 22.044 & 109.146 & 7.243 & 2897.560 & 0.0149 \\
 22 & 50.258 & 140.727 & 7.409 & 3586.010 & 0.0145 \\
 24 & 67.049 & 175.481 & 7.826 & 4508.310 & 0.0141 \\
 26 & 112.924 & 219.255 & 9.509 & 6428.150 & 0.0025 \\
 28 & 247.583 & 295.533 & 11.133 & 8728.530 & 0.0004 \\
 30 & 284.252 & 307.698 & 10.679 & 9611.320 & 0.0006 \\
 40 & 296.633 & 330.824 & 11.294 & 18070.600 & 0.0004 \\
 50 & 807.522 & 368.500 & 16.900 & 42250.000 &  0.00001 \\ \hline
\end{tabular}
\end{center}
\caption{Comparison between the on-the-fly algorithm and the iterative method.}
\label{tab:allpairs}
\end{table}

First, we compare the execution times of the on-the-fly algorithm with those of the iterative method
proposed in Section~\ref{sec:iterativemethod}. 
Since the iterative method only allows for the computation of the distance for all state pairs at once,
the comparison is (in fairness) made with respect to runs of our on-the-fly algorithm with input 
query being the set of all state pairs. For each input instance, the comparison involves the following steps:
\begin{enumerate}[label=\({\alph*}]
  \item we run the on-the-fly algorithm, storing both execution time and the 
  number of solved transportation problems,
  \item then, on the same instance, we execute the iterative method until the running 
  time exceeds that of step 1. We report the number of iterations and
  the number of solved transportation problems.
\item Finally, we calculate the approximation error between the exact solution $\dist$ computed 
  by our method at step 1 and the approximate result $d$ obtained in step 2 by the iterative method, 
  as $\norm{\dist}{d}$.
\end{enumerate}
This has been made on a collection of CTMCs varying from $10$ to $50$ states. 
For each $n = 10, \dots, 30$, we have considered $40$ randomly generated CTMCs 
per out-degree, varying from $3$ to $n$; whereas for $n = 40$ and $50$, the out-degree varies from $3$ to $10$.
Table~\ref{tab:allpairs} reports the average results of the comparison obtained for a discount factor $\lambda = \frac{1}{2}$.

As it can be seen, our use of a greedy strategy in the construction of the couplings leads to a significant improvement in the performances. We are able to compute the exact solution before the iterative method can under-approximate it with an absolute error of $\approx 0.03$, which is a non-negligible error for a value within the interval $[0,1]$. 

\begin{table}[t]
\begin{center}
\footnotesize
\setlength{\tabcolsep}{1.5ex}
\begin{tabular}[c]{|c|c|c|c|c|}
\hline
 \multirow{2}{*}{\# States} & \multicolumn{2}{c|}{out-deg = 3} &  
 \multicolumn{2}{c|}{ $3 \leq \text{out-deg} \leq \text{\# States} / 2$ } \\ \cline{2-5}
 & Time (s) & \# TPs & Time (s) & \# TPs \\ \hline\hline 
 30 & 0.304 & 0.383 & 18.113 & 21.379 \\
 40 & 2.045 & 0.954 & 34.582 & 22.877 \\
 50 & 7.832 & 16.304 & 50.258 & 139.427 \\ \hline
\end{tabular}
\qquad
\begin{tabular}[c]{|c|c|c|}
\hline
 \multirow{2}{*}{\# States} & \multicolumn{2}{c|}{out-deg = 3} \\ \cline{2-3}
  & Time (s) & \# TPs  \\ \hline\hline 
 60 & 34.858 & 12.053  \\
 70 & 48.016 & 14.166  \\
 80 & 73.419 & 29.383  \\
 90 & 75.591 & 13.116  \\
 100 & 158.027 & 20.301 \\ \hline
\end{tabular}
\end{center}

\caption{Average performances of the on-the-fly algorithm on single-pair queries. 
Execution times and number of performed TPs are reported for CTMCs with different out-degree. 
For instances with more than 50 states the out-degree is fixed to $3$;}
\label{tab:singlepairs}
\end{table}

So far, we only examined the case when the on-the-fly algorithm is run on all state pairs at once.
Now, we show how the performance of our method is improved even further when the distance
is computed only for single pairs of states.
Table~\ref{tab:singlepairs} shows the average execution times and number of solved 
transportation problems for (nontrivial) single-pair queries for randomly generated 
of CTMCs with number of states varying from 30 to 100. In the first two columns we consider
CTMCs with out-degree equal to 3, while the last two columns show the average values 
for out-degrees varying from 3 to haft of the number of states of the CTMCs. 
The results show that, when the out-degree of the CTMCs is low, our algorithm performs 
orders of magnitude better than in the general case. 

Notably, our on-the-fly method scales well when the out-degree is small and successive computation of the current $\lambda$-discrepancy are performed on a relatively small set of pairs.

As for the linear program characterization of the bisimilarity distance illustrated in Section~\ref{sec:LPcharact}, tests performed on small CTMCs show that solving $D_\lambda(\M)$ is inefficient in practice, both using the simplex and the interior-point methods\footnote{The implementation is done in $\text{Wolfram Mathematica}^{\circledR}\, 9$ and uses the Linear Program solvers available in the standard library.}. Even for CTMCs with less than $20$ states, the computation times are in the order of hours. For this reason, the efficiency of our on-the-fly technique is by no mean comparable to the linear program solution.

\section{Conclusions and Future Work} \label{sec:conclusions}
In this paper, we proposed a bisimilarity pseudometric for measuring the behavioral similarity between CTMCs, that extends that on MCs introduced by Desharnais et al. in \cite{DesharnaisGJP04}. Moreover, we gave a novel linear program characterization of the distance that, differently from similar previous proposals, have a number of constraints which is polynomial in the size of the CTMC. This proved that the bisimilarity pseudometric can be computed in polynomial time. 
Finally, we defined an on-the-fly algorithm for computing the bisimilarity distance. We demonstrated that, using on-the-fly techniques the computation time is improved with orders of magnitude with respect to the corresponding iterative and linear program approaches. Moreover, our technique allows for the computation on a set of target distances that might be done by only investigating a significantly reduced set of states, and for further improvement of speed.

Our algorithm can be practically used to address a large spectrum of problems. For instance, it can be seen as a method to decide whether two states of a given CTMC are probabilistic bisimilar, to identify bisimilarity classes, or to solve lumpability problems. 
It is sufficiently robust to be used with approximation techniques as, for instance, to provide a least over-approximation of the behavioral distance given over-estimates of some particular distances. It can be integrated with other approximate algorithms, having the advantage of the efficient on-the-fly state space exploration.

Having a practically efficient tool to compute bisimilarity distances opens the perspective of new applications already announced in previous research papers. One of these is the state space reduction problem for CTMCs. Our technique can be used in this context as an indicator for the sets of neighbour states that can be collapsed due to their similarity; it also provides a tool to estimate the difference between the initial CTMC and the reduced one, hence a tool for the approximation theory of CTMCs.

\section*{Acknowledgement}
We would like to thank the anonymous reviewers that with their suggestions greatly improved the presentation of the paper. In particular, we thank Franck van Breugel for the helpful discussions and for providing us the omitted proofs in the proceeding publication of~\cite{ChenBW12}.


\bibliographystyle{alpha}
\bibliography{biblio}

\end{document}